\documentclass{article}
\newcommand{\titel}{Conditional logic as a short-circuit logic} 

\usepackage[all]{xy}
\usepackage{amssymb}
\usepackage{stmaryrd}
\usepackage{amsmath,amsthm} 
\usepackage{graphicx}
\usepackage{url}
\usepackage[hidelinks]{hyperref} 

\usepackage{environ}
\makeatletter
\NewEnviron{Lalign}{\tagsleft@true\begin{align}\BODY\end{align}}
\makeatother

\usepackage{enumerate,tikz}
\usepackage{xcolor}

\newcommand{\si}{s}

\newcommand{\BF}{\ensuremath{\textit{BF}_A}}
\newcommand{\BFu}{\ensuremath{\textit{BF}_A^{\:\und}}}
\newcommand{\mBF}{\ensuremath{\textit{MBF}_A}}
\newcommand{\mBFu}{\ensuremath{\textit{MBF}_A^{\,\und}}}
\newcommand{\baf}{\ensuremath{\mathit{bf}}}
\newcommand{\memf}{\ensuremath{\mathit{mf}}}
\newcommand{\membf}{\ensuremath{\mathit{mbf}}}

\newcommand{\qedex}{\textit{End~example.}}
\newcommand{\axname}[1]{\textup{\ensuremath{\textrm{#1}}}}
\newcommand{\SCL}{\axname{SCL}}
\newcommand{\export}{\mathbin{\setlength{\unitlength}{1ex}
     \begin{picture}(2.0,1.8)(-.8,0)
     \put(-.5,1.6){\line(1,0){1.4}}
     \put(-.5,-0.2){\line(1,0){1.4}}
     \put(-.44,-0.2){\line(0,1){1.8}}
     \put(.84,-0.2){\line(0,1){1.8}}
     \end{picture}
     }}
\newcommand{\CP}{\axname{CP}}
\newcommand{\CPu}{\ensuremath{\CP^\und}}

\newcommand{\CL}{\axname{CL}}

\newcommand{\FSCL}{\axname{FSCL}}
\newcommand{\MSCL}{\axname{MSCL}}
\newcommand{\SSCL}{\axname{SSCL}}

\newcommand{\CLSCL}{\axname{C$\ell$SCL}}
\newcommand{\CLSCLtwo}{\axname{C$\ell$SCL$_2$}}
\newcommand{\CLSCLu}{\ensuremath{\CLSCL}}
\newcommand{\CLtwo}{\ensuremath{\CL_2}} 
\newcommand{\CLthree}{\ensuremath{\CL_3}} 
\newcommand{\CLu}{\ensuremath{\CL^\und}} 

\newcommand{\Le}{\ell}
\newcommand{\Ri}{r}
\newcommand{\ri}{\Ri}

\colorlet{darkblue}{blue!80!black}
\newcommand{\blauw}[1]{\textcolor{darkblue}{#1}}

\newcommand{\sha}{\ensuremath{\widetilde{\alpha}}}

\newcommand{\CPmem}{\ensuremath{\axname{\CP}_{\mem}}}
\newcommand{\CPmemu}{\ensuremath{\CPmem^\und}}
\newcommand{\mem}{\ensuremath{\textit{mem}}}
\newcommand{\CPcond}{\ensuremath{\axname{\CP}_{\cond}}}
\newcommand{\cond}{\ensuremath{\textit{cl}}}
\newcommand{\CPcondu}{\ensuremath{\axname{\CP}_{\cond}^\und}}

\newcommand{\PS}{\ensuremath{{\mathcal{C}_A}}}
\newcommand{\PSu}{\ensuremath{{\mathcal{C}_A^\und}}}
\newcommand{\SP}{\ensuremath{{\mathcal{S}_A}}}
\newcommand{\SPu}{\ensuremath{{\mathcal{S}_A^\und}}}
\newcommand{\SigCP}{\ensuremath{\Sigma_\CP(A)}}
\newcommand{\SigCPu}{\ensuremath{\Sigma_\CP^\und(A)}}
\newcommand{\SigSCL}{\ensuremath{\Sigma_\SCL(A)}}
\newcommand{\SigSCLu}{\ensuremath{\Sigma_\SCL^\und(A)}}
\newcommand{\TCP}{\ensuremath{\mathbb{T}_{{\SigCP},\cal X}}}
\newcommand{\TSCL}{\ensuremath{\mathbb{T}_{{\SigSCL},\cal X}}}

\newcommand{\MSCLe}{\axname{EqMSCL}}
\newcommand{\CLe}{\axname{EqCL}}
\newcommand{\CLeu}{\ensuremath{\CLe^\und}}
\newcommand{\MSCLeu}{\ensuremath{\axname{EqMSCL}^\und}}
\newcommand{\MSCLu}{\ensuremath{\axname{MSCL}^\und}}
\newcommand{\FSCLu}{\ensuremath{\axname{FSCL}^\und}}

\newcommand{\true}{\ensuremath{\textit{true}}}
\newcommand{\false}{\ensuremath{\textit{false}}}
\newcommand{\lef}{\ensuremath{\scalebox{0.78}{\raisebox{.1pt}[0pt][0pt]{$\;\lhd\;$}}}}
\newcommand{\rig}{\ensuremath{\scalebox{0.78}{\raisebox{.1pt}[0pt][0pt]{$\;\rhd\;$}}}}

\newcommand{\tr}{\ensuremath{{\sf T}}}
\newcommand{\fa}{\ensuremath{{\sf F}}}
\newcommand{\und}{\ensuremath{{\sf U}}}
\newcommand{\undefi}{\ensuremath{\textit{undefined}}}

\newtheorem{theorem}{Theorem}[section]
\newtheorem{lemma}[theorem]{Lemma}  
\newtheorem{corollary}[theorem]{Corollary}  
\newtheorem{definition}[theorem]{Definition}  
\theoremstyle{definition}
\newtheorem{example}[theorem]{Example}

\newcommand{\SB}{\ensuremath{\mathbb{S}_3}} 
\newcommand{\SBtwo}{\ensuremath{\mathbb{S}_2}} 

\newcommand{\leftand}{~
     \mathbin{\setlength{\unitlength}{1ex}
     \begin{picture}(1.4,1.8)(-.3,0)
     \put(-.6,0){$\wedge$}
     \put(-.54,-0.2){\textcolor{white}{\circle*{0.6}}}
     \put(-.54,-0.2){\circle{0.6}}
     \end{picture}
     }}
     
\newcommand{\fulland}{~
     \mathbin{\setlength{\unitlength}{1ex}
     \begin{picture}(1.4,1.8)(-.3,0)
     \put(-.6,0){$\wedge$}
     \put(-.54,-0.2){\circle*{0.6}}
     \end{picture}
     }}

\newcommand{\leftor}{~
     \mathbin{\setlength{\unitlength}{1ex}
     \begin{picture}(1.4,1.8)(-.3,0)
     \put(-.6,0){$\vee$}
     \put(-.54,1.54){\textcolor{white}{\circle*{0.6}}}
     \put(-.54,1.54){\circle{0.6}}
     \end{picture}
     }}

\newcommand{\fullor}{~
     \mathbin{\setlength{\unitlength}{1ex}
     \begin{picture}(1.4,1.8)(-.3,0)
     \put(-.6,0){$\vee$}
     \put(-.54,1.54){\circle*{0.6}}
     \end{picture}
     }}
     
\newcommand{\fullimp}{~
     \mathbin{\setlength{\unitlength}{1ex}
     \begin{picture}(1.8,1.8)
     \put(-.0,0){$\rightarrow$}
     \put(-.1,0.57){\circle*{0.6}}
     \end{picture}
     ~}}

\renewenvironment{quote}{%
  \list{}{%
    \leftmargin0.2cm   
    \rightmargin\leftmargin
  }
  \item\relax
}
{\endlist}
\begin{document}
\title{\titel}
\date{}

\author{
\small Jan A. Bergstra \& Alban Ponse\\
\small Informatics Institute, University of Amsterdam,\\
\small Science Park 900, 1098 XH, Amsterdam, The Netherlands\\
\small j.a.bergstra@uva.nl 
 \& a.ponse@uva.nl
}

\maketitle

\begin{abstract}
Three-valued conditional logic (CL) is 
defined by Guzmán and Squier (1990), and based on McCarthy's
noncommutative connectives,
axiomatises a short-circuit logic (SCL)
that defines more identities than three-valued MSCL (Memorising SCL, which also has a two-valued variant). 
This follows from the fact that the definable connective that prescribes
full left-sequential conjunction is commutative in CL.
We show that in CL, the full left-sequential connectives and negation define Bochvar’s three-valued strict logic.

We observe that CL also has a two-valued variant of which 
the full left-sequential connectives and negation 
define a commutative logic that is weaker than propositional logic because the absorption laws do not hold. 

Next, we show that the original, equational axiomatisation of CL is not independent and
give several alternative, independent axiomatisations.
\end{abstract}

\noindent
\textbf{Keywords \& phrases:}
Conditional logic, ~short-circuit evaluation, ~short-circuit logic, 
\\
left-sequential connectives, ~Bochvar's logic

{\small\tableofcontents}

\newpage

\section{Introduction}

\begin{table}
{\small
\centering
\rule{1\textwidth}{.4pt}
\\[-5mm]
\begin{minipage}[t]{0.47\linewidth}\centering
\begin{Lalign}
\label{a}
\tag{1.1a}
&x'' = x
\\
\label{b}
\tag{1.1b}
&(x \wedge y)' = x' \vee y'
\\
\label{c}
\tag{1.1c}
&(x \wedge y) \wedge z = x\wedge (y \wedge z)
\\
\label{d}
\tag{1.1d}
&x\wedge (y \vee z) = (x \wedge y) \vee (x \wedge z)\\[-2mm]\nonumber
\end{Lalign}
\end{minipage}
\begin{minipage}[t]{0.52\linewidth}\centering
\begin{Lalign}
\label{e}
\tag{1.1e}
&(x \vee y) \wedge z = (x \wedge z) \vee (x' \wedge y \wedge z)
\\
\label{f}
\tag{1.1f}
&x \vee (x \wedge y) = x
\\
\label{g}
\tag{1.1g}
&(x \wedge y) \vee (y \wedge x) = (y \wedge x) \vee (x \wedge y)
\end{Lalign}
\end{minipage}
\hrule
}
\caption{The set of axioms of CL, given in~\cite{GS90}}
\label{tab:CL}
\end{table}

Conditional logic (CL) is defined by 
Guzmán and Squier in the paper \emph{The algebra of conditional logic}~\cite{GS90}, where
the conjunction adapts McCarthy's
conjunction to the domain of the three truth values \true, \false, \undefi\ with constants 
\tr, \fa, \und, respectively.
In CL, the 3-element algebra $C = \{\tr, \fa, \und\}$ is studied with the operations
$x \mapsto x':C\to C$ and $(x, y)\mapsto x \wedge y, x \vee y : C^2\to C$ given by the tables
\[
\begin{array}[t]{l|l}
&'\\
\hline\\[-3mm]
\tr&\fa\\
\fa&\tr\\
\und&\und
\end{array}
\hspace{2cm}
\begin{array}[t]{l|lll}
\wedge&\tr&\fa&\und\\
\hline\\[-3mm]
\tr&\tr&\fa&\und\\
\fa&\fa&\fa&\fa\\
\und&\und&\und&\und
\end{array}
\hspace{2cm}
\begin{array}[t]{l|lll}
\vee&\tr&\fa&\und\\
\hline\\[-3mm]
\tr&\tr&\tr&\tr\\
\fa&\tr&\fa&\und\\
\und&\und&\und&\und
\end{array}
\]
A \emph{$C$-algebra} is defined as an algebra with the operations $'$, $\wedge$ and $\vee$. 
This set-up does not require that the constants are defined in a $C$-algebra, and thus includes
four basic cases:
\\[-6mm] 
\begin{enumerate}[(i)]
 \setlength\itemsep{0mm}
\item 
no distinguished elements, 
\item \und\ distinguished, 
\item \tr\ and \fa\ distinguished, and
\item all elements of $C$ distinguished.
\end{enumerate}
In~\cite{GS90} it is shown that the laws in Table~\ref{tab:CL} are complete 
for the variety of $C$-algebras.
Laws~\eqref{a} and~\eqref{b} imply that this variety satisfies a 
duality principle identical to the one in ordinary Boolean algebra.
For the case that the constants for these three truth values are distinguished, the 
following three axioms are added:
\[
\und'=\und, \quad \tr\wedge x=x, \quad \tr'=\fa.
\]
We end this brief introduction to CL with a quote from~\cite{GS90} about the name ``conditional logic'' 
and ``short-circuit evaluation'' :
{\small\begin{quote}
``For example, let $B=\{T, F\}$ with ordinary negation ($'$), 
conjunction ($\wedge$) and disjunction ($\vee$).
[...]
\\
Up to anti-isomorphism, there is a unique regular extension of $B$ to a 3-valued logic with 
non-commutative $\wedge$ and $\vee$ which satisfies deMorgan's laws $x'' = x$ and $(x \wedge y)'= x' \vee y'$. 
Following~\cite{Gries81}, we call this algebra ``conditional logic.'' Conditional logic was first studied by 
McCarthy~\cite{MC63a} and~\cite{MC63}. It is the logic of choice in several programming languages (C, Prolog, 
Lisp, ...) in which the idea of ``short circuit evaluation'' is implemented.
It leads to faster evaluation of the logical expression, since evaluation 
stops as soon as an answer can be obtained.'' \blauw{[Reference numbering has been adjusted.]}
\end{quote}}

We first explain why we view short-circuit evaluation (SCE) as the basic evaluation strategy 
in any  setting that \emph{prescribes} sequential evaluation of $\wedge$ and $\vee$.
Following \cite{BergstraBR1995}, we write 
\[
x\leftand y
\]
for the conjunction of $x$ and $y$ that {prescribes}
SCE, where the small circle indicates that the left argument must be evaluated first, and 
use the name \emph{short-circuit conjunction}. 
We write $\leftor$ 
for \emph{short-circuit disjunction}, which as above is defined dually: 
\[
x\leftor y=\neg(\neg x\leftand\neg y).
\]

From now on, we will mostly use $\leftand$ and $\leftor$ and $\neg$ when writing about \CL. 
We write
\[
\CLtwo
\quad\text{ and }\quad
\CLthree
\]
respectively for the case where constants \tr\ and \fa\  are
distinguished and the two defining axioms $\neg\tr=\fa$ and $\tr\leftand x=x$ are added to those 
of $\CL$, and the case where \und\ is also distinguished and the axiom $\neg\und=\und$ is added. 
If no constants are distinguished, we simple write \CL.

Another well-known sequential evaluation strategy is \emph{full left-sequential evaluation}, 
where both arguments of a conjunction
and a disjunction are always evaluated. The connective \emph{full left-sequential conjunction}, denoted by
$x\fulland y$, 
prescribes full left-sequential evaluation and can be defined in terms of SCE:
\[x\fulland y=(x\leftor (y\leftand\fa))\leftand y.
\]
Thus, if $x$ evaluates to \true, $y$ is evaluated and determines the evaluation value, and if
$x$ evaluates to \false, $y\leftand\fa$ is evaluated, which determines that
$x\fulland y$ has the value \false. The connective \emph{full left-sequential disjunction},
denoted by $x\fullor y$, is defined dually: $x\fullor y=\neg(\neg x\fulland \neg y)$.

Short-circuit evaluation admits several semantics (valuation congruences), which prompted the 
definition of several equational logics that we introduced in~\cite{BPS13} as two-valued
\emph{short-circuit logics} (SCLs). For example, $x\leftand x=x$ is refuted in  
Free SCL (\FSCL), while it holds in Memorising SCL (\MSCL). 
Next, Static SCL (\SSCL) is the short-circuit logic that represents propositional logic 
with connectives ${\leftand}$ and ${\leftor}$ instead of $\wedge$ and $\vee$, and thus with short-circuit
evaluation as the prescribed evaluation strategy.
For example, $x\leftand y=y\leftand x$ is refuted in  
\MSCL, while it holds in
in \SSCL.\footnote{%
  In~\cite{BPS13}, five different, two-valued SCLs are defined, of which \FSCL\ is the weakest and
  \SSCL, a sequential version of propositional logic, is the strongest.}

These SCLs can be ordered by a chain $\FSCL\prec\MSCL\prec\SSCL$:
\begin{align*}
\FSCL\vdash s=t&~\Longrightarrow~ \mathcal L\vdash s=t
\quad\text{with $\mathcal L$ any SCL},\\
\MSCL\vdash s=t&~\Longrightarrow~ \SSCL\vdash s=t.
\end{align*}
The logic \FSCL\ is immune to side effects: in FSCL, $x\leftand x=x$ 
does not hold because a side effect occurring at the first evaluation of an atom 
(propositional variable) 
may alter its second 
evaluation result, while \MSCL\ represents the setting in which
the evaluation of each atom in an expression is memorised, so that
(atomic) side effects cannot occur.
In \MSCL, the entire process of evaluation is distinctive and therefore ${\leftand}$ is not commutative,
e.g., $x\leftand\fa$ requires evaluation of $x$, while $\fa\leftand x$ does not.
Examples of laws of \FSCL\ (and thus of \MSCL) are 
\[\text{$\neg\neg x=x, 
~~\tr\leftand x =x, 
~~\fa\leftand x=\fa$ ~~\text{and}
~ 
$(x\leftand y)\leftand z=x\leftand(y\leftand z)$.}
\]
In \MSCL, the connective ${\fulland}$ does not depend on the presence of the constant
\fa\ because 
\[
\MSCL\vdash x\fulland y =(x\leftand y)\leftor(y\leftand x).
\]

We observe that \CLtwo\ 
can be seen as a two-valued variant of \CL\ that axiomatises
a short-circuit logic that resides in between \MSCL\ and \SSCL, which we will define as 
\CLSCLtwo\ (two-valued Conditional SCL).
A simple and interesting example that distinguishes \MSCL\ and \CLSCLtwo\ is the equation
\[
x\fulland y=y\fulland x,
\] 
which holds in the latter and is refuted in \MSCL.
Because \CLSCLtwo\ refutes $x\leftand y=y\leftand x$, it follows that
$\MSCL\prec\CLSCLtwo\prec\SSCL$.

In~\cite{BPS21} we defined three-valued versions of \FSCL\ and \MSCL, denoted \FSCLu\ and \MSCLu, respectively.
\CLthree\ axiomatises \CLSCL, the three-valued version of \CLSCLtwo, 
and $\MSCLu\prec\CLSCLu$ follows from the same example.
(Of course, \SSCL\ has no extension with \und\ because commutativity of ${\leftand}$ implies 
$\und = \und\leftand\fa=\fa\leftand\und=\fa$.)

We continue with a brief introduction to the building blocks for our generic definition 
of short-circuit logics.
In 1985, Hoare introduced in~\cite{Hoa85} the \emph{conditional}, a ternary connective with 
notation $x\lef y\rig z$ that models \textbf{if $y$ then $x$ else $z$}, and 
provided eleven equational axioms to show that the conditional and two constants for 
truth and falsehood characterise the propositional calculus.
With the conditional, algebraic properties 
can be elegantly expressed, e.g., $\tr \lef x\rig \fa=x$
and $x \lef (y \lef z \rig u)\rig v = (x \lef y \rig v) \lef z \rig (x \lef u \rig v)$. 
However, the same result was proved in 1948 by Church 
in~\cite{Chu48}, where he introduced \emph{conditioned disjunction},\footnote{%
  For the conditioned disjunction, reference~\cite{Chu56} (1956)
  is often used, and also the name \emph{conditional disjunction}.} 
notation $[x,y,z]$, that also models \textbf{if $y$ then $x$ else $z$}.
  Moreover, in~\cite{Chu48}, the duality property of $[x,y,z]$ 
is highlighted and even older related work is mentioned, a quote: 
  \begin{quote}{\small
  ``\emph{Conditioned disjunction, $t$, and $f$, are 
  a complete set of independent primitive connectives for the propositional calculus.} This has
  been proved by E.\ L. Post as a part of his comprehensive study of primitive connectives for the
  propositional calculus. [footnote: E. L. Post. 
  \emph{The Two-Valued Iterative Systems of Mathematical Logic} (Annals
  of Mathematics Studies, no.\ 5), Princeton, N. J., 1941.] But it seems worth while to give a separate 
  proof here. [footnote: The method of the proof of completeness is that used by E. L. Post in 1921 in connection
  with a different set of primitive connectives. 
  See the \emph{American Journal of Mathematics}, vol. 43, pp. 167-168.]''
  \blauw{[Our comment: 
  the mentioned ``method of the proof'' is a simple induction with respect to the number of variables 
  of a possible connective.]}
}
  \end{quote}
Both~\cite{Hoa85} and~\cite{Chu48} do not address the fact that the ternary connective defined therein 
in fact prescribes short-circuit evaluation.

In~\cite{BergstraP2011}, we introduced \emph{proposition algebra}, motivated 
by the observation that Hoare's conditional both naturally prescribes short-circuit evaluation 
and admits several semantics (valuation congruences), the least distinguishing of which is that 
of sequential propositional logic. 
Furthermore,
\[x\leftand y = y\lef x\rig\fa
\]
can be seen as a formal definition of what is usually described as 
the short-circuit evaluation of conjunction.
In~\cite{BPS13}, we used proposition algebra as an algebraic framework to define several short-circuit logics,
including \FSCL, \MSCL\ and \SSCL, and in~\cite{BPS21} 
to define \FSCLu\ and \MSCLu, i.e. the three-valued variants of \FSCL\ and \MSCL.

\noindent\textbf{Contents of the paper.}~
\\
-
In Section~\ref{sec:propa}, we recall the definitions of  
the proposition algebras \CP\ and \CPmem\ that underlie the definitions of \FSCL\ and \MSCL, respectively.
We also recall the three-valued variants.
\\
- 
In Section~\ref{sec:3}, we define the two-valued proposition algebra \CPcond\ which is related to 
conditional logic.
We define and explain in detail so-called `CL-basic forms', based on the mem-basic forms for \CPmem, which 
resemble the ``reduced trees'' defined in~\cite{GS90} that provide a 
``standard model'' for free $C$-algebras.
We prove that \CLtwo\ 
completely axiomatises \CLSCLtwo, the SCL defined by \CPcond\ and the short-circuit connectives.
\\
-
In Section~\ref{sec:4}, we define the three-valued proposition algebra \CPcondu\ and do the same as in 
Section~\ref{sec:3} for the three-valued variants.
It follows that \CLthree\ is a complete axiomatisation of \CLSCLu.
\\
-
In Section~\ref{sec:indep}, we show that the axiomatisation of CL in Table~\ref{tab:CL} is not
independent and provide several axiomatisations of \CL, \CLtwo, \CLu\ (i.e. \CL\ with only \und), 
and \CLthree\ that are independent.
\\
-
In Section~\ref{sec:conc}, 
we consider the distinguishing example
$x\fulland y=y\fulland x$ that proves $\MSCLu\prec \CLSCLu$  
and show that in \CLSCLu, and thus in \CLthree,
the (definable) full left-sequential connectives and negation
define Bochvar’s three-valued strict logic.
We conclude with a comment on the role of the constants \tr\ and \fa\ in SCLs and some remarks on 
related and future work.

\section{Proposition algebra and short-circuit logics}
\label{sec:propa}
In this section we recall the definitions of two-valued and three-valued proposition algebra, and 
of the short-circuit logics mentioned in the previous section.

\begin{table}
{
\centering
\rule{1\textwidth}{.4pt}
\begin{align*}
\label{cp1}
\tag{CP1} 
\CP:&
&
x \lef \tr \rig y &= x\\
\label{cp2}\tag{CP2}
&&x \lef \fa \rig y &= y\\
\label{cp3}\tag{CP3}
&&\tr \lef x \rig \fa  &= x\\
\label{cp4}\tag{CP4}
&&x \lef (y \lef z \rig u)\rig v &= 
	(x \lef y \rig v) \lef z \rig (x \lef u \rig v)
\\[0mm]
\cline{1-4}
\\[-4mm]
\label{cpu}
\tag{CP-U}
\CPu:&
&
x\lef\und\rig y
 &=\und
\end{align*}
\hrule
}
\caption{\CP, a set of basic axioms of proposition algebra, 
and $\CPu=\CP\cup \eqref{cpu}$}
\label{tab:CP}
\end{table}

\begin{definition}
\label{def:CP}
For $A$ a set of atoms (propositional variables), define the signatures 
\[
\SigCP=\{\_\lef\_\rig\_,\tr,\fa,a\mid a\in A\}
\quad\text{and}\quad
\SigCPu=\SigCP\cup\{\und\}\]
and the sets $\PS$ and $\PSu$ of \textbf{conditional expressions} 
as the set of closed terms over \SigCP\ and \SigCPu, respectively.
\end{definition}

In Table~\ref{tab:CP} we display \CP, a set of basic axioms for the conditional 
defined in~\cite{BergstraP2011}, and \CPu, which was defined in~\cite{BPS21}.
The dual of a conditional expression $P$, notation $P^d$, is defined by $\tr^d=\fa$, $\fa^d=\tr$,
$\und^d=\und$,
$a^d=a$ for $A\in A$,
and $(P\lef Q\rig R)^d=R^d\lef Q^d\rig P^d$. Taking the dual 
of a variable $x$ as itself, it easily follows that \CP\ and \CPu\ are self-dual axiomatisations.

The axioms of \CPu\ define a congruence over \PSu: 
\emph{free valuation congruence}, which can be characterised
with help of basic forms.\footnote{%
  About terminology: 1) Two-valued free valuation congruence (FVC) was defined in~\cite{BergstraP2011}
  in terms of ``valuation algebras'', hence the name.
  2)
  We use ``basic form'' rather than ``normal form'' because 
  the ``natural'' normal form of $\tr\lef a\rig\fa$ would be $a$, while the basic form 
  of atom $a$ is $\tr\lef a\rig\fa$.}

\begin{definition}
\label{def:basic}
\textbf{Basic forms}
are defined by the following grammar
\[t::= \tr\mid\fa\mid\und\mid t\lef a \rig t\quad\text{for $a\in A$.}\]
We write \BFu\ for the set of basic forms over $A$, and $\BF$ for the subset of
\BFu\ in which \und\ does not occur. 
The constans in a basic form $P\in\BFu$ are called  \textbf{leaves}. 
The \textbf{depth} $d(P)$ of $P\in\BFu$ is defined by 
$d(\tr) = d(\fa) = d(\und) = 0$ and 
$d(Q \lef a \rig R) = 1 + \max\{d(Q), d(R)\}$. 

\end{definition}

In~\cite[Appendix A8]{BPS21} we prove that for each $P\in\PSu$ there is a unique $Q\in\BFu$ such that
$\CPu\vdash P=Q$. 
This implies that for all $P,Q\in\PSu, ~\CPu\vdash P=Q$ if, and only if, $P$ and $Q$ have the same (unique)
basic form, that is, if $P$ and $Q$ are free valuation congruent.

\emph{Memorising valuation congruence} is
defined by \CPmemu, i.e., \CPu\ extended with the axiom
\begin{align}
\label{CPmem}
\tag{CPmem}
x\lef y\rig(z\lef u\rig (v\lef y\rig w))=x\lef y\rig(z\lef u\rig w),
\end{align}
which expresses that the evaluation of $y$ is memorised. Furthermore, in~\cite{BPS13},
\CPmem\ is defined as the extension of \CP\ with~\eqref{CPmem}.

By replacing $y$ 
by $\fa\lef y\rig\tr$ and $u$ by $\fa\lef u\rig\tr$, 
it follows that
\[\CPmemu\vdash ((w\lef y\rig v)\lef u\rig z)\lef y\rig z =(w\lef u\rig z)\lef y\rig x,\]
which is the dual of the above axiom, so the duality principle holds in \CPmemu.

\begin{definition}
\label{def:membasic}
The constants \tr, \fa\ and \und\ are \textbf{mem-basic forms}, and
a basic form $P\lef a\rig Q$ is a 
\textbf{mem-basic form} if $a$ does not occur in $P$ and $Q$.

\noindent
We write \mBFu\ for the set of basic forms over $A$, and $\mBF$ for the subset of
\mBFu\ in which \und\ does not occur.
\end{definition}

\begin{definition}
\label{def:aux}
For $a\in A$, the auxiliary functions $\Le_a,\Ri_a:\BFu\to\BFu$, ``left $a$-reduction'' and ``right $a$-reduction'',
are defined by 
\begin{align*}
&\text{$\Le_a(\tr)=\tr,~~\Le_a(\fa)=\fa,~~\Le_a(\und)=\und$,}
&&\text{$\Ri_a(\tr)=\tr, ~~\Ri_a(\fa)=\fa,~~\Ri_a(\und)=\und$,}\\[2mm]
&\Le_a(P\lef b\rig Q)=
\begin{cases}
P&\text{if $b=a$},\\
\Le_a(P)\lef b\rig \Le_a(Q)&\text{if $b\ne a$},
\end{cases}
\quad 
&&
\Ri_a(P\lef b\rig Q)=\begin{cases}
Q&\text{if $b=a$},\\
\Ri_a(P)\lef b\rig \Ri_a(Q)&\text{if $b\ne a$}.
\end{cases}
\end{align*}
\end{definition}

So, for each $P\in\mBFu$, both $\Le_a(P)$ and $\Ri_a(P)$ are mem-basic forms that do not contain $a$.
Mem-basic forms and the functions $\Le_a()$ and $\Ri_a()$ are a point of departure in the next section.

{In~\cite[Thm.5.7]{BP17} we proved a result that implies that for each $P\in\PS$, 
there is a unique $Q\in\mBF$ such that $\CPmem\vdash P=Q$,
and for all $P,Q\in\PS, ~\CPmem\vdash P=Q$ if, and only if, $P$ and $Q$ have the same (unique)
mem-basic form, that is, if $P$ and $Q$ are memorising valuation congruent.
Given the approach in~\cite{BPS21}, 
the same results follow for the three-valued setting.
}

\begin{definition}
\label{def:SigSCL}
For $A$ a set of atoms (propositional variables), define the signatures 
\[
\SigSCL=\{\leftand, \leftor, \neg,\tr,\fa, a\mid a \in A\}
\quad\text{and}\quad
\SigSCLu=\SigSCL\cup\{\und\}
\]
and the sets $\SP$ and $\SPu$ of \textbf{sequential expressions} 
as the set of closed terms over $\SigSCL$ and \SigSCLu, respectively.
\end{definition}

The short-circuit connectives are defined in \CPu\ (and in \CP) by the following axioms:
\begin{align}
\label{defneg}
\tag{defNeg}
\neg x &=\fa\lef x\rig \tr,\\
\label{defand}
\tag{defAnd}
x\leftand y &=y\lef x\rig\fa,
\end{align}
and, by duality, short-circuit disjunction should satisfy in \CPu\ (and in \CP)
\begin{equation}
\label{defor}
\tag{defOr}
x\leftor y=\tr\lef x\rig y,
\end{equation}
which we show below.
In a similar way it can be shown that full left-sequential conjunction, which we already defined by 
$x\fulland y=(x\leftor (y\leftand\fa))\leftand y$, 
satisfies 
\[
x\fulland y=y\lef x\rig(\fa\lef y\rig\fa).
\]

In~\cite{BPS13} we defined (two-valued) short-circuit logics in a generic way with help of the 
conditional by restricting the consequences of a particular \CP-axiomatisation extended with 
equations~\eqref{defneg} and~\eqref{defand} to the signature $\SigSCL$, 
and we repeat here some of these definitions.
So, the conditional connective is considered a hidden operator.
In these definitions, the export operator $\export$ of Module algebra~\cite{BHK90}
is used to express hiding in a concise way: 
in module algebra, $S\export X$  is the operation that 
exports the signature $S$ from module $X$ while declaring 
other signature elements hidden. 

\begin{definition}
\label{def:SCL}
A \textbf{short-circuit logic}
is a logic that implies the consequences
of the module expression
\begin{align*}
\SCL=\{\tr,\neg,\leftand\}\export(&\CP\cup\{\eqref{defneg},\eqref{defand}\}).
\end{align*}
\end{definition}

\begin{definition}
\label{def:FSCL}
\textbf{Free short-circuit logic $(\FSCL)$}
is the short-circuit logic that implies no other 
consequences than those of the module expression \SCL.

\noindent
\textbf{Memorising short-circuit logic $(\MSCL)$}
is the short-circuit logic that implies no other 
consequences than those of the module expression 
\[
\{\tr,\neg,\leftand\}\export(\CPmem\cup\{\eqref{defneg},\eqref{defand}\}).
\]
\textbf{Static short-circuit logic $(\SSCL)$}
is the short-circuit logic that implies no other 
consequences than those of the module expression 
\[
\{\tr,\neg,\leftand\}\export(\CPmem\cup\{\fa\lef x\rig\fa=\fa, \eqref{defneg},\eqref{defand}\}).
\]
\end{definition}

To enhance readability, we extend these short-circuit logics with the constant \fa\ and its 
defining axiom $\fa=\neg\tr$,
which is justified by the \SCL-derivation
\begin{align*}
\fa&=\fa\lef \tr\rig\tr
&&\text{by~\eqref{cp1}}\\
&=\neg\tr,
&&\text{by~\eqref{defneg}}
\end{align*}
and with the connective $\leftor$ and its defining 
axiom $x\leftor y=\neg(\neg x\leftand\neg y)$, which is justified by 
\begin{align*}
\neg(\neg x\leftand\neg y)
&=\fa\lef((\fa\lef y\rig\tr)\lef(\fa\lef x\rig\tr)\rig\fa)\rig\tr
&&\text{by~\eqref{defneg}, \eqref{defand}}\\
&=\fa\lef(\fa\lef x\rig(\fa\lef y\rig\tr))\rig\tr
&&\text{by~\eqref{cp4}, \eqref{cp2}, \eqref{cp1}}\\
&=(\fa\lef\fa\rig\tr)\lef x\rig(\fa\lef(\fa\lef y\rig\tr)\rig\tr)
&&\text{by~\eqref{cp4}}\\
&=\tr\lef x\rig y.
&&\text{by~\eqref{cp2}, \eqref{cp4}, \eqref{cp1}}
\end{align*}

\begin{table}
{
\centering
\rule{1\textwidth}{.4pt}
\begin{align}
\label{Neg}
\tag{Neg}
\fa&=\neg\tr
\\
\label{Or}
\tag{Or}
x\leftor y
&=\neg(\neg x\leftand\neg y)
\\ 
\label{Tand}
\tag{Tand}
\tr\leftand x&=x
\\
\label{Abs}
\tag{Abs}
x\leftand(x\leftor y)
&=x
\\
\label{Mem}
\tag{Mem}
(x\leftor y)\leftand z
&=(\neg x\leftand(y\leftand z))\leftor(x\leftand z)
\hspace{2cm}
\end{align}
\hrule
}
\caption{The set \MSCLe\ of axioms for \MSCL}
\label{tab:MSCLe}
\end{table}

The equational logic $\MSCLe$ is defined by the five axioms in Table~\ref{tab:MSCLe}.
In~\cite{BPS21} it is proved that \MSCLe\ axiomatises \MSCL.

Some consequences of $\MSCLe$ are these (see~\cite{BPS21}):
\\[-6mm]
\begin{enumerate}[(i)]
 \setlength\itemsep{0mm}
\item 
the connective $\leftand$ is associative, i.e. $(x\leftand y)\leftand z=x\leftand(y\leftand z)$,
\item 
the connective $\leftand$ is idempotent, i.e. $x\leftand x=x$,
\item
the connective $\leftand$ is left-distributive, i.e. 
$x\leftand (y\leftor z)=(x\leftand y)\leftor (x\leftand z)$, and
\item
$\neg\neg x=x$, ~$x\leftand \neg x=\neg x\leftand x$, and $x\leftand(y\leftand x)=x\leftand y$.
\end{enumerate} 
~\\[-6mm]
An important
property of \MSCL\ is that the conditional can be expressed:
\begin{align*}
(x\leftand y)\leftor(\neg x\leftand z)
&=\tr\lef(y\lef x\rig\fa)\rig(\fa\lef x\rig z)
&&\text{by~\eqref{defneg}, \eqref{defor}, \eqref{cp2}, \eqref{cp4}}\\
&=(\tr\lef y\rig(\fa\lef x\rig z))\lef x\rig
(\fa\lef x\rig z)
&&\text{by \eqref{cp4}, \eqref{cp2}}\\
&=(\tr\lef y\rig\fa)\lef x\rig (u\lef \fa\rig(\fa\lef x\rig z))
&&\text{by \eqref{CPmem}, \eqref{cp2}}\\
&=y\lef x\rig z.
&&\text{by \eqref{cp3}, \eqref{CPmem}, \eqref{cp2}}
\end{align*}

In~\cite{BPS21}, we proved
$\MSCLe\vdash(y\leftand x)\leftor (\neg y\leftand z)
=(\neg y\leftand z)\leftor(y\leftand x)
=(y\leftor z)\leftand (\neg y\leftor x)$.

Finally, we defined in~\cite{BPS21} the short-circuit logics \FSCLu\ and \MSCLu\
by  using \CPu\ and \CPmemu, respectively, instead of their two-valued variants, and  exporting \und.
We proved that \MSCL\ and \MSCLu\ are axiomatised by the equational logics \MSCLe\ and $\MSCLeu=
\MSCLe\cup\{\neg\und=\und\}$, respectively.

\section{Conditional short-circuit logic, the two-valued case}
\label{sec:3}

In this section we restrict to the two-valued setting. In particular, where we write about
mem-basic forms, we refer to those in \mBF\ (see Definition~\ref{def:membasic}).
We define `CL-basic forms' based on the mem-basic forms for \CPmem\ and the short-circuit logic \CLSCLtwo\,
and prove that \CLtwo\ 
completely axiomatises \CLSCLtwo, the two-valued SCL defined by \CPcond\ and the short-circuit connectives.

\medskip

Consider the axiom
\begin{align*}
\label{eq:c1}
\tag{CL1}
\tr\lef (y\lef x\rig \fa)\rig(x\lef y\rig \fa) = \tr\lef (x\lef y\rig \fa)\rig(y\lef x\rig \fa),
\end{align*}
i.e. the translation of the \CL-axiom $(x\leftand y)\leftor(y\leftand x)=(y\leftand x)\leftor (x\leftand y)$
to \CP, and
the following axiom that is more convenient for a semantics based on 
mem-basic forms:
\begin{equation}
\label{eq:clcond}
\tag{CL2}
(x\lef y\rig z)\lef u\rig(v\lef y\rig w) = (x\lef u\rig v)\lef y\rig(z\lef u\rig w).
\end{equation}

\begin{lemma}
\label{la:convenient}
$(i)~
\CPmem + \eqref{eq:c1}\vdash \eqref{eq:clcond}$, and 
$(ii)~
\CPmem + \eqref{eq:clcond}\vdash
\eqref{eq:c1}.$
\end{lemma}

\begin{proof}
$(i)$~ Derive
\begin{align}
\label{eq:8}
x\lef &(\tr\lef(y\lef u\rig\fa)\rig (u\lef y\rig \fa)) \rig v\\
\nonumber
&=
x\lef (y\lef u\rig\fa)\rig (x\lef (u\lef y\rig \fa) \rig v)\\
\nonumber
&=
(x\lef y\rig(x\lef (u\lef y\rig \fa) \rig v))\lef u\rig(x\lef (u\lef y\rig \fa) \rig v)\\
\nonumber
&=
(x\lef y\rig((x\lef u\rig v)\lef y\rig  v))\lef u\rig((x\lef u\rig v)\lef y\rig  v)\\
&=
\label{eq:9}
(x\lef y\rig v)\lef u\rig(v\lef y\rig  v), 
\end{align}
and apply axiom~\eqref{eq:c1} in~\eqref{eq:8}, so exchange $u$ and $y$. By~\eqref{eq:9}, 
\[
(x\lef y\rig v)\lef u\rig(v\lef y\rig  v)=(x\lef u\rig v)\lef y\rig(v\lef u\rig  v).
\]
Now replace $v$ by $v\lef y\rig (z\lef u\rig w)$ and apply memorisation:
\[
(x\lef y\rig z)\lef u\rig(v\lef y\rig  w)=(x\lef u\rig v)\lef y\rig(z\lef u\rig  w).
\]
$(ii)$~ Derive $\tr\lef(y\lef x\rig\fa)\rig(x\lef y\rig\fa)=(\tr\lef y\rig\fa)\lef x\rig (\fa\lef y\rig\fa)$
and apply~\eqref{eq:clcond}.
\end{proof}

\begin{definition}
\label{CPcond}
Let \CPcond\ be the extension of \CPmem\ with axiom \eqref{eq:clcond}.
\textbf{Conditional short-circuit logic $(\CLSCLtwo)$}
is the short-circuit logic that implies no other 
consequences than those of the module expression 
\[
\{\tr,\neg,\leftand\}\export(\CPcond\cup\{\eqref{defneg},\eqref{defand}\}).
\]
\end{definition}

Observe that the axiom~\eqref{eq:clcond} is self-dual, so the duality principle holds in \CPcond.
In the setting of CL, we shall adopt a strict total ordering $<$ of $A$, in order to 
deal with the axiom~\eqref{eq:clcond}.
Whenever we want to be explicit about this ordering, we write 
\[
(A,<).
\]
Below we define CL-basic forms and prove that $\CPcond\vdash P=Q$ if, and only if, $P$ and $Q$
are \emph{conditional valuation congruent}, that is, if $P$ and $Q$
have the same unique CL-basic form (relative to a strict total ordering of $A$). 

\begin{example}
\label{ex:trees}
Consider the mem-basic form $\textit{Input}=P\lef a\rig Q$ with
\begin{align*}
P&=((\tr\lef d\rig\fa)\lef c\rig(\tr\lef d\rig\fa))\lef b\rig((\tr\lef c\rig\fa)\lef d\rig\tr),
\\
Q&=(\tr\lef d\rig\fa)\lef c\rig(\tr\lef d\rig\fa).
\end{align*}
Below, \textit{Input} is shown graphically as the tree at the top left. 
Furthermore, $\CPcond\,$-equivalent representations of \textit{Input}, with root atoms $a$ and $d$, 
respectively, are shown in the right column, and another with root atom $a$ and $d<c$ 
is shown at the bottom left:
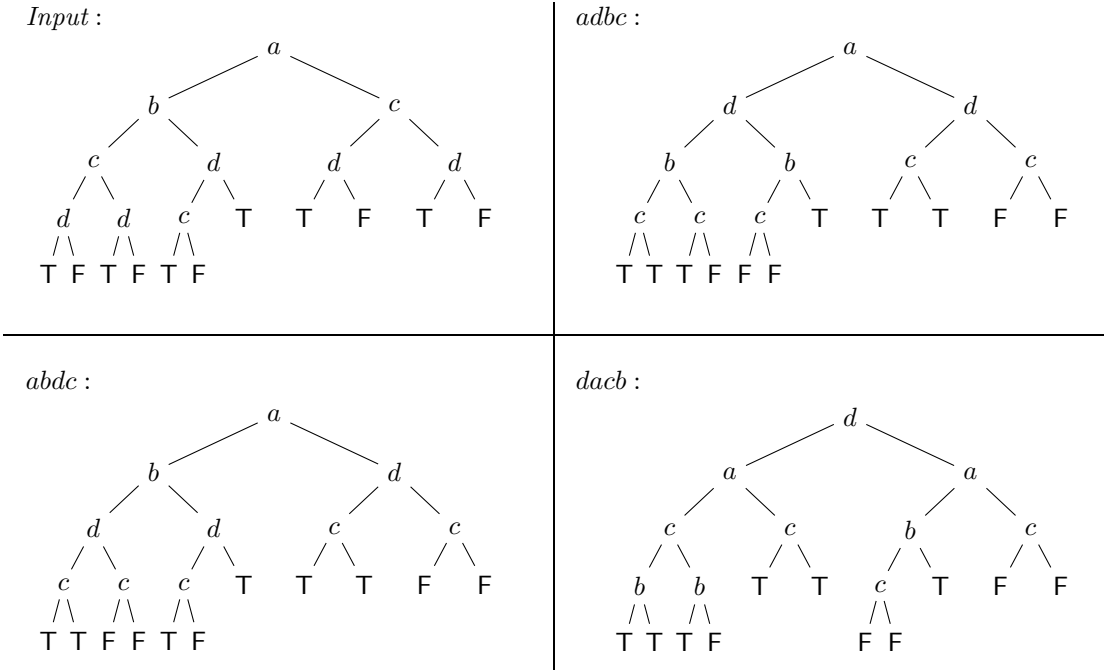
\[
\begin{array}{l|l}
~\textit{Input}:&~adbc:\\
\begin{array}{l}
\begin{tikzpicture}[%
level distance=7.5mm,
level 1/.style={sibling distance=32mm},
level 2/.style={sibling distance=16mm},
level 3/.style={sibling distance=8mm},
level 4/.style={sibling distance=4mm}
]
\node (root) {$a$}
  child {node (lef1) {$b$}
    child {node (lef2) {$c$}
      child {node (lef3r) {$d$}
        child {node (lef4l) {$\tr$}} 
        child {node (lef4r) {$\fa$}}
        } 
      child {node (lef3r) {$d$}
        child {node (lef4l) {$\tr$}} 
        child {node (lef4r) {$\fa$}}
        } 
    }
    child {node (rig2) {$d$}
      child {node (lef3r) {$c$}
        child {node (lef4l) {$\tr$}} 
        child {node (lef4r) {$\fa$}}
        } 
      child {node (rig3r) {$\tr$}}
    }
  }
  child {node (rig1) {$c$}
    child {node (lef2) {$d$}
      child {node (lef3) {$\tr$}} 
      child {node (rig3) {$\fa$}}
    }
    child {node (rig2) {$d$}
      child {node (lef3r) {$\tr$}} 
      child {node (rig3r) {$\fa$}}
    }
  };
\end{tikzpicture}
\end{array}
\quad
&
\quad
\begin{array}{l}
\begin{tikzpicture}[%
level distance=7.5mm,
level 1/.style={sibling distance=32mm},
level 2/.style={sibling distance=16mm},
level 3/.style={sibling distance=8mm},
level 4/.style={sibling distance=4mm}
]
\node (root) {$a$}
  child {node (lef1) {$d$}
    child {node (lef2) {$b$}
      child {node (lef3r) {$c$}
        child {node (lef4l) {$\tr$}} 
        child {node (lef4r) {$\tr$}}
        } 
      child {node (lef3r) {$c$}
        child {node (lef4l) {$\tr$}} 
        child {node (lef4r) {$\fa$}}
        } 
    }
    child {node (rig2) {$b$}
      child {node (lef3r) {$c$}
        child {node (lef4l) {$\fa$}} 
        child {node (lef4r) {$\fa$}}
        } 
      child {node (rig3r) {$\tr$}}
    }
  }
  child {node (rig1) {$d$}
    child {node (lef2) {$c$}
      child {node (lef3) {$\tr$}} 
      child {node (rig3) {$\tr$}}
    }
    child {node (rig2) {$c$}
      child {node (lef3r) {$\fa$}} 
      child {node (rig3r) {$\fa$}}
    }
  };
\end{tikzpicture}
\\
\end{array}
\\[20mm]
\hline
\\[0mm]
~abdc:&~dacb:\\
\begin{array}{l}
\begin{tikzpicture}[%
level distance=7.5mm,
level 1/.style={sibling distance=32mm},
level 2/.style={sibling distance=16mm},
level 3/.style={sibling distance=8mm},
level 4/.style={sibling distance=4mm}
]
\node (root) {$a$}
  child {node (lef1) {$b$}
    child {node (lef2) {$d$}
      child {node (lef3r) {$c$}
        child {node (lef4l) {$\tr$}} 
        child {node (lef4r) {$\tr$}}
        } 
      child {node (lef3r) {$c$}
        child {node (lef4l) {$\fa$}} 
        child {node (lef4r) {$\fa$}}
        } 
    }
    child {node (rig2) {$d$}
      child {node (lef3r) {$c$}
        child {node (lef4l) {$\tr$}} 
        child {node (lef4r) {$\fa$}}
        } 
      child {node (rig3r) {$\tr$}}
    }
  }
  child {node (rig1) {$d$}
    child {node (lef2) {$c$}
      child {node (lef3) {$\tr$}} 
      child {node (rig3) {$\tr$}}
    }
    child {node (rig2) {$c$}
      child {node (lef3r) {$\fa$}} 
      child {node (rig3r) {$\fa$}}
    }
  };
\end{tikzpicture}
\end{array}
\quad
&
\quad
\begin{array}{l}
\begin{tikzpicture}[%
level distance=7.5mm,
level 1/.style={sibling distance=32mm},
level 2/.style={sibling distance=16mm},
level 3/.style={sibling distance=8mm},
level 4/.style={sibling distance=4mm}
]
\node (root) {$d$}
  child {node (lef1) {$a$}
    child {node (lef2) {$c$}
      child {node (lef3r) {$b$}
        child {node (lef4l) {$\tr$}} 
        child {node (lef4r) {$\tr$}}
        } 
      child {node (lef3r) {$b$}
        child {node (lef4l) {$\tr$}} 
        child {node (lef4r) {$\fa$}}
        } 
    }
    child {node (rig2) {$c$}
      child {node (lef3r) {$\tr$}}
      child {node (rig3r) {$\tr$}}
    }
  }
  child {node (rig1) {$a$}
    child {node (lef2) {$b$}
      child {node (lef3) {$c$}
        child {node (AA) {$\fa$}}
        child {node (BB) {$\fa$}}
        } 
      child {node (rig3) {$\tr$}}
    }
    child {node (rig2) {$c$}
      child {node (lef3r) {$\fa$}} 
      child {node (rig3r) {$\fa$}}
    }
  };
\end{tikzpicture}
\\
\end{array}\end{array}
\]
This example is a first step to our definition of so-called \emph{CL-basic forms},
a proper subset of the mem-basic forms.
\hfill\qedex
\end{example}

We discuss the question of how mem-basic forms can be assigned a 
unique representation modulo conditional valuation congruence, where this representation itself must of 
course be a mem-basic form.
In the following recursive definition, the ``shared alphabet'' of a mem-basic form 
is the set of atoms that occur in \emph{each} evaluation path
(whereas the ``alphabet'' $\alpha(P)$ of a mem-basic form $P$ is the set of all its atoms, i.e. 
those that occur in at least one evaluation path).

\begin{definition}
For a mem-basic form $P$, its 
\textbf{shared alphabet} $\sha(P)$
is the set defined  by
\[
\sha(\tr)=\sha(\fa)=\emptyset, \quad\sha(P_1\lef a\rig P_2)
=\{a\}\cup(\sha(P_1))\cap\sha(P_2)).
\]
\end{definition}

Hence, $\sha(P)=\emptyset$ iff $P\in\{\tr,\fa\}$, and if $P=P_1\lef a\rig P_2$, then $a\in\sha(P)$. 
In Example~\ref{ex:trees}, $\sha(\textit{Input})=\{a,d\}$, 
and this is also the shared alphabet of the other mem-basic forms. 

\begin{lemma}
\label{la:sha}
If $P$ is a mem-basic form and $a\in\sha(P)$, then $\CPcond\vdash P=\Le_a(P)\lef a\rig \Ri_a(P)$.
\end{lemma}

\begin{proof}
By structural induction on $P$. 
If $P=P_1\lef a\rig P_2$, then $\Le_a(P)=P_1$ and $\Ri_a(P)=P_2$, so we are done, and
if $P=P_1\lef b\rig P_2$ with $b\ne a$, then $a\in(\sha(P_1)\cap\sha(P_2))$, hence 
\begin{align*}
\CPcond\vdash 
&=P_1\lef b\rig P_2
\\
&=(\Le_a(P_1)\lef a\rig\Ri_a(P_1))\lef b\rig(\Le_a(P_2)\lef a\rig\Ri_a(P_2))
&&\text{(by induction)}\\
&=(\Le_a(P_1)\lef b\rig\Le_a(P_2))\lef a\rig(\Ri_a(P_1)\lef b\rig\Ri_a(P_2))
&&\text{(by axiom~\eqref{eq:clcond})}\\
&=\Le_a(P)\lef a\rig\Ri_a(P).
\end{align*}
\end{proof}

Lemma~\ref{la:sha} implies the following consequence, which shows that $\sha()$ is preserved under
provable equality: 

\begin{lemma}
\label{la:cta}
If $P, Q$ are mem-basic forms and $\CPcond\vdash P=Q$, then $\sha(P)=\sha(Q)$.
\end{lemma}

\begin{proof}
If $a\in\sha(P)$, then $\CPcond\vdash P=\Le_a(P)\lef a\rig\Ri_a(P)$, so 
$\CPcond\vdash Q=\Le_a(P)\lef a\rig\Ri_a(P)$ and thus $a\in\sha(Q)$. By symmetry, $\sha(P)=\sha(Q)$.
\end{proof}

To decide whether different mem-basic forms are provably equal in \CPcond, we define so-called
CL-basic forms in a recursive way.
We first define some technical notions that are used in these basic forms.

\begin{definition} 
For $A$ a nonempty set of atoms, $A^\si$ is the set of words over $A$ with the property that each 
$\sigma\in A^\si$ contains no multiple occurrences of the same atom. We write $\epsilon$
for the empty string, thus $\epsilon\in A^s$. Moreover, we define $\emptyset^\si=\{\epsilon\}$.

For $\sigma\in A^\si$ with length $|\sigma|=n$ and $P_0,...,P_{2^n-1}$ mem-basic forms, define $F_\sigma(P_0,...,P_{2^n-1})$ by
\begin{align*}
&\sigma=\epsilon:
&&F_{\sigma}(P_0)=P_0,
\\
&\sigma=a\rho, ~a\in A:
&&F_{\sigma}(P_0,...,P_{2^{n}-1})=
F_{\rho}(P_0,...,P_{2^{|\rho|}-1})\lef a\rig F_{\rho}(P_{2^{|\rho|}},...,P_{2^{n}-1}).
\end{align*}
We say that
$\sigma\in A^\si$ \textbf{agrees with} $(A,<)$ if $\sigma$'s atoms occur in increasing order.
\end{definition}

\begin{definition}
\label{def:clbasic}
A mem-basic form $P\in\mBF$ is a \textbf{CL-basic form} if 
for the unique $\sigma\in A^\si$ of length $|\sha(P)|$ 
that agrees with $(A,<)$,
\[\text{$P=F_{\sigma}(P_{0},..., P_{2^{|\sigma|}-1})$ and each
$P_i$ is a CL-basic form.} 
\]
\end{definition}

Observe that a mem-basic form $P$ is a CL-basic form if axiom~\eqref{eq:clcond} cannot be applied to $P$
(i.e. if $|\sha(Q)|\leq 1$ for all subterms $Q$ of $P$), e.g. as in 
\(
(\tr\lef b\rig\fa)\lef a\rig(\tr\lef c\rig\fa).
\)
\paragraph{Notation and convention.} For $A'$ a finite set of atoms 
we can \emph{represent $(A',<)$} by the string $\rho\in (A')^\si$ that contains all atoms of $A'$ 
in increasing order (thus $|\rho|=|A'|$). 
We will often just write \textbf{Cbf} instead of ``CL-basic form''.

In Example~\ref{ex:trees}, the terms displayed in the right column are Cbfs if  
$adbc$ agrees with $(A,<)$ or if $dacb$ does, respectively, 
and those in the left column are not 
because they do not satisfy the condition $P=F_{ad}(...)$ or $P=F_{da}(...)$, respectively. 
We give another example that further illustrates the role of $(A,<)$ and provides more intuition 
about the existence of unique Cbfs.

\begin{example}
\label{ex:trees2}
Consider the following mem-basic form $P$, which we display graphically:
\\[3mm]
\indent
\qquad
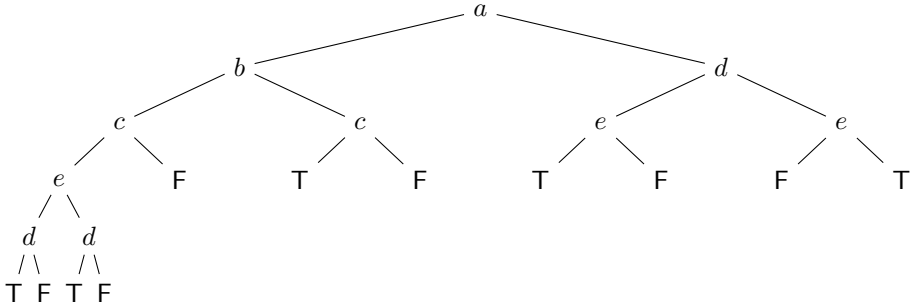
\begin{tikzpicture}[%
level distance=7.5mm,
level 1/.style={sibling distance=64mm},
level 2/.style={sibling distance=32mm},
level 3/.style={sibling distance=16mm},
level 4/.style={sibling distance=8mm},
level 5/.style={sibling distance=4mm}
]
\node (root) {$a$}
  child {node (lef1) {$b$}
    child {node (cld) {$c$}
      child {node (lef3r) {$e$}
        child {node (lef4l) {$d$}
          child {node (lef 51) {$\tr$}}
          child {node (lef 52) {$\fa$}}
          }
        child {node (lef4r) {$d$}
          child {node (lef 51) {$\tr$}}
          child {node (lef 52) {$\fa$}}
          }
        } 
      child {node (lef3r) {$\fa$}}
    }
    child {node (crd) {$c$}
      child {node (lef3r) {$\tr$}}
      child {node (rig3r) {$\fa$}}
    }
  }
  child {node (rig1) {$d$}
    child {node (lef2) {$e$}
      child {node (lef3) {$\tr$}} 
      child {node (rig3) {$\fa$}}
    }
    child {node (rig2) {$e$}
      child {node (lef3r) {$\fa$}} 
      child {node (rig3r) {$\tr$}}
    }
  };
\end{tikzpicture}
\\[1mm]
Clearly, $P$ is a mem-basic form and $\sha(P)=\{a\}$. 
Observe that some subterms have overlap and can be rearranged according to axiom~\eqref{eq:clcond}.
To obtain a unique CL-basic form, the two outer subterms $(\tr\lef d\rig\fa)\lef e\rig (\tr\lef d\rig\fa)$ 
and $(\tr\lef e\rig\fa)\lef d\rig (\tr\lef e\rig\fa)$
should be made uniform and it must be decided whether $b$ and the two adjacent $c$'s 
should be interchanged (with rearranged subterms) or not.

If $abcde$ agrees with $(A,<)$, the left subterm $(\tr\lef d\rig\fa)\lef e\rig (\tr\lef d\rig\fa)$ should be replaced by
\[(\tr\lef e\rig\tr)\lef d\rig (\fa\lef e\rig\fa),\] 
and if $acbde$ agrees with $(A,<)$, the left subterm $Q\lef b\rig R$ should be replaced by one with root $c$
that we display graphically:

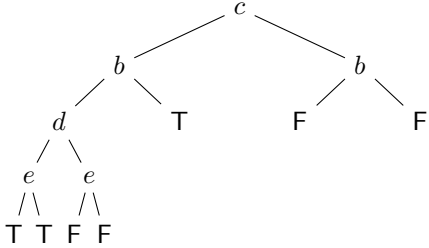
\qquad
\begin{tikzpicture}[%
level distance=7.5mm,
level 1/.style={sibling distance=32mm},
level 2/.style={sibling distance=16mm},
level 3/.style={sibling distance=8mm},
level 4/.style={sibling distance=4mm}
]
\node (root) {$c$}
  child {node (lef1) {$b$}
    child {node (lef2) {$d$}
      child {node (lef3r) {$e$}
        child {node (lef4l) {$\tr$}} 
        child {node (lef4r) {$\tr$}}
        }  
      child {node (lef3r) {$e$}
        child {node (lef4l) {$\fa$}} 
        child {node (lef4r) {$\fa$}}
        } 
    }
    child {node (rig2) {$\tr$}}
  }
  child {node (rig1) {$b$}
    child {node (lef2) {$\fa$}}
    child {node (rig2) {$\fa$}}
  };
\end{tikzpicture}

Note that each alternative ordering $(A,<')$ determines how $P$ should be transformed into a provably 
equal Cbf because it determines whether $b<'c$ or $c<'b$ and $d<'e$ or $e<'d$, and
that $a$ is the only possible root atom of such a transformation.
\hfill\qedex
\end{example}

For $P$ a mem-basic form different from \tr\ or \fa,
a consequence of axiom~\eqref{eq:clcond} concerns a provably equal representation that is ordered
using $\sha(P)$.

\begin{lemma}
\label{la:markings}
For each mem-basic form $P\in\mBF$ 
and $\sigma\in(\sha(P))^\si$ that has maximal length and agrees with $(A,<)$,
there exist mem-basic forms $P_0,...,P_{2^{|\sigma|}-1}$ such that
\[
\CPcond\vdash P=F_\sigma(P_0,...,P_{2^{|\sigma|}-1}).
\]
\end{lemma}

\begin{proof}
By induction on the depth of $P$. The base case $d(P)=0$ is trivial. 
If $d(P)\geq1$ and
$\sigma=a\rho$, then $\rho$ represents the ordering of the shared alphabets of
the mem-basic forms $\Le_{a}(P)$ and $\Ri_{a}(P)$ that
are both of lesser depth, hence
\begin{align*}
\CPcond\vdash P&=\Le_{a}(P)\lef a\rig\Ri_{a}(P)
&&\text{(by Lemma~\ref{la:sha})}\\
&=F_{\rho}(P_0,...,P_{2^{|\rho|}-1})\lef a\rig F_{\rho}(P_0',...,P_{2^{|\rho|}-1}')
&&\text{(IH)}\\
&=F_{\sigma}(P_0,...,P_{2^{|\sigma|}-1})
\end{align*}
if $P_{2^{|\rho|}+i}=P_i'$ for $i=0,...,2^{|\rho|}-1$.
\end{proof}

\begin{theorem}
\label{thm:zak3}
For each mem-basic form $P\in\mBF$ there is a unique 
CL-basic form $Q$ such that $\CPcond\vdash P=Q$.
\end{theorem}

\begin{proof}
By induction on the depth of $P$. 
If $d(P)=0$ then $P\in\{\tr,\fa\}$ and we are done.

If $d(P)>0$ then $|\sha(P)|=n$ for some $n\geq1$. 
For the unique $\sigma\in(\sha(P))^\si$ of length $n$ that agrees with $(A,<)$
and Lemma~\ref{la:markings},
$\CPcond\vdash P=F_\sigma(P_0,...,P_{2^n-1})$. 
By induction, there are unique Cbfs $Q_i$ such that $\CPcond\vdash P_i=Q_i$. 
Hence, $Q=F_\sigma(Q_0,...,Q_{2^n-1})$ is the unique
Cbf that satisfies $\CPcond\vdash P=Q$.
\end{proof}

In~\cite[Thm.5.7]{BP17} we prove that for each $P\in\PS$, there is a unique mem-basic form obtained
by a normalisation function $\membf()$ such that
$\CPmem\vdash P = \membf(P)$. Hence, we find the following result:

\begin{corollary}
\label{cor:cl2}
For all $P,Q\in\PS$, 
$\CPcond\vdash P=Q$
if, and only if, 
\[
\llbracket P\rrbracket_\cond=\llbracket Q\rrbracket_\cond, 
\text{where $\llbracket P\rrbracket_\cond$ is $P$'s  unique CL-basic form}.
\]
\end{corollary}

Corollary~\ref{cor:cl2} and the above theorem imply that 
\CPcond\ axiomatises two-valued
conditional valuation congruence, that is, the congruence defined by having equal CL-basic forms.

It remains to be shown that \CLSCLtwo\ is completely axiomatised by the equational logic \CLtwo. 
Given a signature $\Sigma$, we write 
\(\mathbb T_{\Sigma,\cal X}\)
for the set of open terms over 
$\Sigma$ with variables in $\cal X$. 
As in~\cite{BPS21}, we define the following translation functions:
\\[-6mm]
\begin{description}
\item
[\quad$f:\TSCL\to\TCP$] by 
\begin{align*}
f(B)&=bl~\text{ for }B\in\{\tr,\fa\},
&f(\neg t)&=\fa\lef f(t)\rig\tr,\\
f(a)&=a~\text{ for }a\in A,
&f(t_1\leftand t_2)&=f(t_2)\lef f(t_1)\rig\fa,\\
f(x)&=x~\text{ for }x\in \cal X,
&f(t_1\leftor t_2)&=\tr\lef f(t_1)\rig f(t_2).
\hspace{32mm}
\end{align*}
\item
[\quad$g:\TCP\to\TSCL$] by 
\begin{align*}
g(B)&=B~\text{ for }B\in\{\tr,\fa\},
&g(x)&=x~\text{ for }x\in \cal X,
\\
g(a)&=a~\text{ for }a\in A,
&g(t_1\lef t_2\rig t_3)&=(g(t_2)\leftand g(t_1))\leftor(\neg g(t_2)\leftand g(t_3)).
\end{align*}
\end{description}

\begin{theorem}
\label{thm:CLcomplete}
For all terms $s,t$ over $\SigSCL$, $\CLtwo\vdash s=t~\iff~\CLSCLtwo\vdash s=t$.
\end{theorem}

\begin{proof} 
In~\cite[Thm.6.5]{BPS21} it is proved that for all terms $s,t$ over $\SigSCL$, 
\[\MSCLe\vdash s=t~\iff~\MSCL\vdash s=t\]
and the two crucial intermediate results in this proof, when lifted to the present case, are
\\
1) for all $s,t\in\TCP$, $\CPcond\vdash s=t~\Rightarrow~ \CLtwo\vdash g(s)=g(t)$, and
\\
2) for all $t\in\TSCL$, $\CLtwo\vdash g(f(t))=t$.

For result 1 it must be shown that the axioms of \CLtwo\ are derivable from \CLSCLtwo. 
By Theorem~\ref{thm:indep2}, $\CLtwo\vdash\MSCLe$, so it remains to be shown 
that axiom~\eqref{eq:clcond} is derivable from $\CLtwo$, and this follows from 
Lemma~\ref{la:convenient} because axiom~\eqref{eq:c1} is the
$g()$-translation of axiom~\eqref{g}.
Result 2 holds for \MSCLe, and hence for \CLtwo.
\end{proof}

Finally, we can use CL-basic forms to define ``normal forms'' in the equational theory of \CLtwo:~\footnote{%
   The term ``normal form'' is not used in~\cite{GS90}, but Guzmán's follow-up 
   paper~\cite{Guz94} states that the completeness of a Gentzen system for
   conditional logic
   ``relies heavily on the axiomatization 
   of the equational 
   theory of $C$, and on the normal form in this equational theory; see~\cite{GS90}. 
   The author believes that an understanding of this normal form, as presented in section 3 of~\cite{GS90}, 
   is essential in order to see through the formalisms in the proofs of (2.3), (2.7) and (2.10).''
   }
for $P\in\SP$ we find $f(P)\in\PS$,
and thus, given an ordering $(A,<)$, a unique CL-basic form $\membf(f(P))$ such that 
$\CPcond\vdash f(P)=\membf(f(P))$ and thus
\[\CLtwo\vdash P=g(f(P))=g(\membf(f(P))),\]
so $g(\membf(f(P)))$ can be defined as the normal form of $P$.

\section{Conditional short-circuit logic, the three-valued case}
\label{sec:4}
In this section we consider the three-valued setting and start from mem-basic forms that may contain \und,
thus the set \mBFu\ of mem-basic forms. 
We proceed as in Section~\ref{sec:3} and it will appear that the CL-basic forms with \und\ are a bit more complex.

\medskip

\begin{definition}
\label{CPcondu}
Let \CPcondu\ be the extension of \CPcond\ with axiom \eqref{cpu}, i.e. $x\lef\und\rig y=\und$.
\\[1mm]
\textbf{Conditional short-circuit logic with undefinedness $(\CLSCLu)$}
is the short-circuit logic that implies no other 
consequences than those of the module expression 
\[
\{\tr,\und,\neg,\leftand\}\export(\CPcondu\cup\{\eqref{defneg},\eqref{defand}\}).
\]
\end{definition}

A typical derivation and a typical example in \CPcondu:
\begin{align*}
\CPcondu\vdash\und\lef x\rig\und
&=(y\lef \und\rig z)\lef x\rig(v\lef\und\rig w)\\
&=(y\lef x\rig v)\lef\und\rig(z\lef x\rig w)\\
&=\und,
\\[2mm]
\CPcondu\vdash (\fa\lef a\rig\und)\lef b\rig(\tr\lef a\rig\und)
&= (\fa\lef b\rig \tr)\lef a\rig(\und\lef b\rig\und)\\
&=(\fa\lef b\rig \tr)\lef a\rig\und.
\end{align*}
This example motivates the following definitions of the extension of $\sha(P)$ to the three-valued case
and of the CL-basic forms.
It will be seen that $(\fa\lef b\rig \tr)\lef a\rig(\und\lef b\rig\und)$ is the basic form according 
to the ordering $a<b$, while this is $(\fa\lef a\rig\und)\lef b\rig(\tr\lef a\rig\und)$ 
if $b<a$. 
Thus, the shared alphabet of \emph{each} of the three terms in this example must be $\{a,b\}$.

\begin{definition}
For $P$ a mem-basic form, 
its \textbf{shared alphabet} $\sha(P)$
is the following set of atoms, defined with help of the auxiliary function 
$f:\mBFu\times 2^A\to 2^A$:
\begin{align*}
&\sha(P)
=f(P,\alpha(P)),\\
&f(\und,x)=x,\\
&f(\tr,x)=f(\fa,x)=\emptyset, \\
&f(P_1\lef a\rig P_2,x)
=\{a\}\cup(f(P_1,x\setminus\{a\})\cap f(P_2,x\setminus\{a\})).
\end{align*}
\end{definition}

For each $\sigma\in A^\si$, we write $F_\sigma(\vec\und)$ for
$F_\sigma(P_0,...,P_{2^{|\sigma|}-1})$ if for all $i$, $P_i=\und$ (where $|\sigma|$ determines the 
length of $\vec\und$).

\begin{lemma}
\label{la:U}
If $P$ is a mem-basic form with only \und-leaves, then (i) $\sha(P)=\alpha(P)$ and 
(ii) for all $\sigma\in A^\si$, $\CPcondu\vdash P=F_\sigma(\vec\und)=\und$.
\end{lemma}

\begin{proof}
Both (i) and (ii) follow by structural induction on $P$.
\end{proof}

\begin{lemma}
\label{la:shaa}
If $P$ is a mem-basic form 
and $a\in\sha(P)$, then 
\(
\CPcondu\vdash P=\Le_a(P)\lef a\rig \Ri_a(P).
\)
\end{lemma}

\begin{proof}
Like the proof of Lemma~\ref{la:sha}.
\end{proof}

We note that the analogue of Lemma~\ref{la:cta} is not valid: 
if $P$ and $Q$ are mem-basic forms and 
$\CPcondu\vdash P=Q$, then $\sha(P)=\sha(Q)$ does not necessarily hold, e.g $\sha(\und)=\emptyset$ and 
$\sha(\und\lef a\rig\und)=\{a\}$.

\begin{definition}
\label{def:clubasic}
A mem-basic form $P$ is a \textbf{CLU-basic form} if 
for the unique $\sigma\in(\sha(P))^\si$ of length $|\sha(P)|$ 
that agrees with $(A,<)$, $P=F_{\sigma}(P_{0},..., P_{2^{|\sigma|}-1})$
and the following conditions are satisfied:
\textup{\begin{enumerate}[(i)]
  \setlength\itemsep{0mm}
\item \emph{for some $j\in \{0,...,2^{|\sigma|}-1\}$, $P_j\ne\und$, and}
\item \emph{for all $i\in \{0,...,2^{|\sigma|}-1\}$,
$P_i$ is a CLU-basic form.}
\end{enumerate}
}
\end{definition}

We will often write \textbf{CUbf} instead of ``CLU-basic form''.
For example, $\und\lef a\rig \tr$ is a CUbf. 
Note that the above definition implies that $F_\rho(\vec\und)$ is not a CUbf for any 
$\rho\in A^\si\setminus\{\epsilon\}$. More generally, a mem-basic form in \mBFu\ of depth larger than 0 with 
only \und-leaves is not a CUbf.

\begin{lemma}
\label{la:markingsu}
If $P$ is a mem-basic form   
and $\sigma\in (\sha(P))^\si$ has maximal length and agrees with $(A,<)$, 
then there exist mem-basic forms $P_0,...,P_{2^{|\sigma|}-1}$ such that
\[\CPcondu\vdash P=F_\sigma(P_0,...,P_{2^{|\sigma|}-1}).\]
\end{lemma}

\begin{proof}
By induction on the depth of $P$. The base case $d(P)=0$ is trivial. 
If $d(P)>0$ and
$\sigma=a\rho$, then $\rho$ represents the ordering of the shared alphabets of
the mem-basic forms $\Le_{a}(P)$ and $\Ri_{a}(P)$ that
are both of lesser depth, hence
\begin{align*}
\CPcondu\vdash P&=\Le_{a}(P)\lef a\rig\Ri_{a}(P)
&&\text{(by Lemma~\ref{la:shaa})}\\
&=F_{\rho}(P_0,...,P_{2^{|\rho|}-1})\lef a\rig F_{\rho}(P_0',...,P_{2^{|\rho|}-1}')
&&\text{(IH)}\\
&=F_{\sigma}(P_0,...,P_{2^{|\sigma|}-1})
\end{align*}
if $P_{2^{|\rho|}+i}=P_i'$ for $i=0,...,2^{|\rho|}-1$.
\end{proof}

\begin{theorem}
\label{thm:zaku}
For each mem-basic form $P$ there is a unique 
CLU-basic form $Q$ such that 
\[\CPcond\vdash P=Q.
\]
\end{theorem}

\begin{proof}
By induction on the depth of $P$. 
If $d(P)=0$ then $P\in\{\tr,\fa,\und\}$ and we are done. 

If $d(P)>0$ then $|\sha(P)|=n$ for some $n\geq1$. 
For the unique $\sigma\in(\sha(P))^\si$ of length $n$ that agrees with $(A,<)$
and Lemma~\ref{la:markingsu},
$\CPcondu\vdash P=F_\sigma(P_0,...,P_{2^n-1})$. 
By induction, there are unique CUbfs $Q_i$ 
such that for all $i$, $\CPcondu\vdash P_i=Q_i$.
The following two cases have to be distinguished:
\begin{enumerate}
  \setlength\itemsep{0mm}
\item For all $i$, $Q_i=\und$. Then $Q=\und$ is a CUbf,
and by Lemma~\ref{la:U}, it is the unique CUbf 
that satisfies $\CPcondu\vdash P=Q$.
\item For some $j\in \{0,...,2^{|\sigma|}-1\}$, $Q_j\ne\und$. 
Then $Q=F_\sigma(Q_0,...,Q_{2^{|\sigma|}-1})$ is a CUbf
and therefore it is the unique CUbf over $(\alpha(P),<)$
that satisfies $\CPcondu\vdash P=Q$.
\end{enumerate}
\end{proof}

\begin{example}
\label{ex:treeU}
Consider the mem-basic form $P$ that we display graphically below on the left, and 
observe that $\sha(P)=\{a,b,c\}$. 
If $\textit{abc}\in A^\si$ agrees with $(A,<)$, the CLU-basic form of $P$ is shown on the right. 
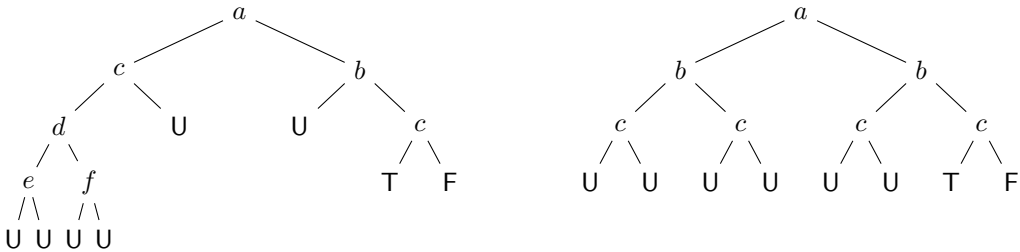
\[
\begin{array}[t]{ll}
\begin{array}[t]{l}
\begin{tikzpicture}[%
level distance=7.5mm,
level 1/.style={sibling distance=32mm},
level 2/.style={sibling distance=16mm},
level 3/.style={sibling distance=8mm},
level 4/.style={sibling distance=4mm}
]
\node (root) {$a$}
  child {node (lef1) {$c$}
    child {node (lef2) {$d$}
      child {node (lef3r) {$e$}
        child {node (lef4l) {$\und$}} 
        child {node (lef4r) {$\und$}}
        }  
      child {node (lef3r) {$f$}
        child {node (lef4l) {$\und$}} 
        child {node (lef4r) {$\und$}}
        } 
    }
    child {node (rig2) {$\und$}}
  }
  child {node (rig1) {$b$}
    child {node (lef2) {$\und$}}
    child {node (rig2) {$c$}
       child {node (XX) {$\tr$}}
       child {node (XX) {$\fa$}}
    }
  };
\end{tikzpicture}
\end{array}
&\qquad
\begin{array}[t]{l}
\\[-35.6mm]
\begin{tikzpicture}[%
level distance=7.5mm,
level 1/.style={sibling distance=32mm},
level 2/.style={sibling distance=16mm},
level 3/.style={sibling distance=8mm},
level 4/.style={sibling distance=4mm}
]
\node (root) {$a$}
  child {node (lef1) {$b$}
    child {node (lef2) {$c$}
      child {node (lef3r) {$\und$}}  
      child {node (lef3r) {$\und$}} 
    }
    child {node (rig2) {$c$}
      child {node (lef3r) {$\und$}}  
      child {node (lef3r) {$\und$}} 
    }
  }
  child {node (rig1) {$b$}
    child {node (lef2) {$c$}
      child {node (lef3r) {$\und$}}  
      child {node (lef3r) {$\und$}} 
    }
    child {node (rig2) {$c$}
       child {node (XX) {$\tr$}}
       child {node (XX) {$\fa$}}
    }
  };
\end{tikzpicture}
\end{array}
\end{array}
\]
Note that the example remains valid
if in both $P$ and its CUbf one of the leaves \tr\ or \fa\ is replaced by \und.
If another order of $\sha(P)$ agrees with $(A,<)$, the corresponding CUbf
is easily found.
\hfill\qedex
\end{example}

The mem-basic form function $\membf: \PS\to\PS$, defined in~\cite{BP17} and used in the previous section, 
can be easily extended to $\PSu\to\PSu$; we show this in the Appendix. 
This implies that for all $P\in\PSu$,
$\CPmemu\vdash P = \membf(P)$, by which we find the following result: 

\begin{corollary}
\label{cor:cl3}
For all $P,Q\in\PSu$, $\CPcondu\vdash P=Q$
if, and only if, 
$\llbracket P\rrbracket_\cond=\llbracket Q\rrbracket_\cond$, 
where $\llbracket P\rrbracket_\cond$ is $P$'s unique CLU-basic form.
\end{corollary}

Corollary~\ref{cor:cl3} and the above theorem imply that 
\CPcondu\ axiomatises three-valued
conditional valuation congruence, that is, the congruence on \PSu\ defined by having equal CLU-basic forms.

It remains to be shown that \CLSCLu\ is completely axiomatised by the equational logic \CLthree. 

\begin{theorem}
\label{thm:CLcompleteu}
For all terms $s,t$ over $\SigSCL$, $\CLthree\vdash s=t~\iff~\CLSCLu\vdash s=t$.
\end{theorem}

\begin{proof} As the proof of Theorem~\ref{thm:CLcomplete}, but now  relying on Theorem~\cite[Thm.7.16]{BPS21}.
The extensions with \und\ in the corresponding intermediate results require $f(\und)=\und$ and $g(\und)=\und$.
\end{proof}

Like in the previous section, we can use CLU-basic forms to define ``normal forms'' in the 
equational theory of \CLthree:
for $P\in\SPu$ we find $f(P)\in\PSu$,
and thus, given an ordering $(A,<)$, a unique CLU-basic form $\membf(f(P))$ such that 
$\CPcondu\vdash f(P)=\membf(f(P))$ and thus
\[\CLthree\vdash P=g(f(P))=g(\membf(f(P))),\]
so $g(\membf(f(P)))$ can be defined as the normal form of $P$.

\section{Independent axiomatisations of CL }
\label{sec:indep}

Guzmán and Squier state in~\cite{GS90} that the seven equational axioms in Table~\ref{tab:CL}, 
i.e.
\\[-4mm]
\begin{minipage}[t]{0.51\linewidth}\centering
\begin{Lalign}
\tag{\ref{a}}
&x'' = x
\\
\tag{\ref{b}}
&(x \wedge y)' = x' \vee y'
\\
\tag{\ref{c}}
&(x \wedge y) \wedge z = x\wedge (y \wedge z)
\\
\tag{\ref{d}}
&x\wedge (y \vee z) = (x \wedge y) \vee (x \wedge z)
\end{Lalign}
\end{minipage}
\begin{minipage}[t]{0.57\linewidth}\centering
\begin{Lalign}
\tag{\ref{e}}
&(x \vee y) \wedge z = (x \wedge z) \vee (x' \wedge y \wedge z)
\\
\tag{\ref{f}}
&x \vee (x \wedge y) = x
\\
\tag{\ref{g}}
&(x \wedge y) \vee (y \wedge x) = (y \wedge x) \vee (x \wedge y)
\end{Lalign}
\end{minipage}
\\[4mm]
are to the best of their knowledge independent. In this section we show that this is not the case
and provide several axiomatisations that are independent. Also, we prove that $\MSCL\prec\CLSCLtwo$.

\medskip

As an alternative to axiom~\eqref{e} that does not assume the associative law~\eqref{c}, we consider both possible
readings of~\eqref{e}:
\begin{Lalign}
\label{e1}
\tag{Mem1}
(x \vee y)\wedge z = (x \wedge z) \vee (x' \wedge (y \wedge z)),
\hspace{52mm}
\\
\label{e2}
\tag{Mem2}
(x \vee y) \wedge z = (x \wedge z) \vee ((x' \wedge y) \wedge z).
\hspace{52mm}
\end{Lalign}
\indent
Furthermore, we shall write \(\CLu\) for \CL\ with (only) the constant \und\ distinguished.
For the axiomatisations of \CLu, \CLtwo\ and \CLthree\ (thus, \und\ or/and  \tr\ and \fa\ are distinguished), 
the following axioms are provided in~\cite{GS90}:
\\[-6mm]
\begin{minipage}[t]{0.32\linewidth}\centering\begin{Lalign}
\label{4a}
\tag{1.4a}
&\und'=\und
\\
\label{4b}
\tag{1.4b}
&\tr\wedge x=x
\\
\label{4c}
\tag{1.4c}
&\tr'=\fa
\end{Lalign}
\end{minipage}

In the proofs below, we used the theorem prover \emph{Prover9} and the tool
\emph{Mace4} for generating finite models, for both tools see~\cite{Prover9}.
Derivations from equational axiomatisations were found with \emph{Prover9}, 
and \emph{Mace4} was used to prove all independence results.
We used these tools on a Macbook Pro with a 2.4GHz 
dual-core Intel Core i5 processor and 4GB of RAM.

\begin{theorem}
\label{thm:indep}
\textup{(i)} Conditional logic with no distinguished constants \textup{(\CL)} is completely axiomatised by
the five axioms \eqref{a}, \eqref{b}, \eqref{e1}, \eqref{f}, \eqref{g}, and this group of axioms is
independent.

\noindent
\textup{(ii)} Conditional logic with no distinguished constants \textup{(\CL)} is completely axiomatised by
the six axioms \eqref{a}, \eqref{b}, \eqref{d}, \eqref{e2}, \eqref{f}, \eqref{g},
and this group of axioms is independent.

\noindent
\textup{(iii)} Conditional logic with (only) \tr\ and \fa\ distinguished \textup{(\CLtwo)} is completely axiomatised by  
axioms \eqref{4b}, \eqref{4c}, and those in
\textup{(i)} or \textup{(ii)}. 
Each of these groups of axioms is independent.

\noindent
\textup{(iv)} Conditional logic with (only) \und\ distinguished \textup{(\CLu)} is completely axiomatised by  
axiom \eqref{4a} and those in
\textup{(i)} or \textup{(ii)}.
Each of these groups of axioms is independent.

\noindent
\textup{(v)} Conditional logic with  \tr, \fa\ and \und\ distinguished \textup{(\CLthree)}  is completely axiomatised by  
axioms \eqref{4a}, \eqref{4b}, \eqref{4c}, and those in
\textup{(i)} or \textup{(ii)}.
Each of these groups of axioms is independent.
\end{theorem}

\begin{proof}
It suffices to show that the axioms in (i) and (ii) follow from \eqref{a}{--}\eqref{g} 
with the appropriate reading of~\eqref{g},
and that both axiomatisations of \CLthree\ are independent.

(i)  
If axiom~\eqref{e} is replaced by~\eqref{e1}, then~\eqref{c} and~\eqref{d} are derivable from the 
mentioned five axioms.
This follows with \emph{Prover9}, the options \texttt{lpo}+\texttt{pass}  
give the best results, 
although the first proof is long and time-consuming 
and both proofs takes less time when \eqref{g} is left out:
first~\eqref{c} and with that, \eqref{d}: 
$< 57''$ and $< 6''$, respectively.

(ii)  
If axiom~\eqref{e} is replaced by~\eqref{e2} then~\eqref{c} is derivable from the 
mentioned six axioms.
This follows with \emph{Prover9}, best result with options \texttt{lpo}+\texttt{unfold} and without~\eqref{g}: $< 36''$. 

Completeness of the axiomatisations in (i){--}(v) follows from the completeness results
proved in~\cite[Cy.2.7]{GS90}.

The two axiomatisations in (v) of \CLthree\ are independent, this quickly follows with \emph{Mace4}.
This implies that all axiomatisations in (i){--}(iv) are also independent.
\end{proof}

The \emph{Prover9} results mentioned above can be obtained slightly faster if \eqref{f} is 
replaced by its dual $x\wedge(x\vee y)=x$ and/or if~\eqref{e1} is replaced by its symmetric counterpart 
$(x \vee y) \wedge z = 
(x' \wedge (y \wedge z))\vee (x \wedge z)$. We note that each of these replacements preserves independence.
Below, we discuss an alternative axiomatisation based on these observations.

We return to our preferred notation $\leftand,\leftor$ and $\neg$, and we further assume 
that where necessary the axioms \eqref{a}{--}\eqref{g} and \eqref{4a}{--}\eqref{4c}
are adapted to this notation.

\begin{table}
{
\centering
\rule{1\textwidth}{.4pt}
\begin{align*}
\tag{\ref{Neg}}
\CLe:&
&\fa&=\neg\tr
\\
\tag{\ref{Or}}
&&x\leftor y
&=\neg(\neg x\leftand\neg y)
\\
\tag{\ref{Tand}}
&&\tr\leftand x&=x
\\
\tag{\ref{Abs}}
&&x\leftand(x\leftor y)
&=x
\\
\tag{\ref{Mem}}
&&(x\leftor y)\leftand z
&=(\neg x\leftand(y\leftand z))\leftor(x\leftand z)
\\
\label{Com}\tag{Com}
&&(x\leftand y)\leftor(y\leftand x)
&=(y\leftand x)\leftor(x\leftand y)
\\[0mm]
\cline{1-4}
\\[-4mm]
\label{Und}
\tag{Und}
\CLeu:&
&\neg\und&=\und
\end{align*}
\hrule
}
\caption{\CLe\ and \CLeu, axiomatisations that are equivalent with \CLtwo\ and \CLthree, respectively}
\label{tab:CL2}
\end{table}

In Table~\ref{tab:CL2} we extend \MSCLe\ to \CLe, and \MSCLeu\ to \CLeu\ with the axiom~\eqref{Com}.
Note the differences between axioms~\eqref{b} and~\eqref{Or}, and 
axioms~\eqref{e1}, \eqref{e2} and~\eqref{Mem}.

\begin{theorem}
\label{thm:indep2}
\textup{(i)} Conditional logic with no constants distinguished \textup{(\CL)} is completely axiomatised by
the five axioms $\neg\neg x=x$ and \eqref{Or}, \eqref{Abs}, \eqref{Mem}, and \eqref{Com} 
from Table~\ref{tab:CL2}.
Moreover, these axioms are independent.

\noindent
\textup{(ii)} Conditional logic with \tr\ and \fa\ distinguished \textup{(\CLtwo)} is completely axiomatised by
the six axioms of \CLe\ (see Table~\ref{tab:CL2}).
Moreover, these axioms are independent.

\noindent
\textup{(iii)} Conditional logic with only \und\ distinguished \textup{(\CLu)} is completely axiomatised by
the six axioms $\neg\neg x=x$, $\neg\und=\und$, and \eqref{Or}, \eqref{Abs}, \eqref{Mem} and \eqref{Com} in 
Table~\ref{tab:CL2}.
Moreover, these axioms are independent.

\noindent
\textup{(iv)} Conditional logic with all constants distinguished \textup{(\CLthree)} is completely axiomatised by
the seven axioms of \CLeu\ (see Table~\ref{tab:CL2}).
Moreover, these axioms are independent.
\end{theorem}

\begin{proof}
It follows easily (either by hand or with help of \emph{Prover9}) that the axioms in each of (i){--}(iv)
both imply those of the associated axiomatisation in 
Theorem~\ref{thm:indep}
that uses the axioms in \ref{thm:indep}.(i), \emph{and} are their consequences. 
This proves the mentioned completeness results.

The independence of the axiomatisations (iii) and (iv) follows quickly with \emph{Mace4},
and implies the independence of (i) and (ii), respectively.
\end{proof}

\begin{corollary}
$\MSCL\prec\CLSCLtwo$ and $\MSCLu\prec\CLSCLu$, i.e., for all terms $s,t$ over \SigSCL,
$\MSCL\vdash s=t ~\Longrightarrow~ \CLSCLtwo\vdash s=t$, and for all terms $s,t$ over \SigSCLu,
\\
\(\MSCLu\vdash s=t ~\Longrightarrow~ \CLSCLu\vdash s=t.\)
\end{corollary}

\begin{proof}
\MSCLe\ axiomatises \MSCL\ and by Theorem~\ref{thm:indep2}.(ii), it is a proper subset of 
the independent axiomatisation \CLe\ of \CLSCLtwo. Hence, $\MSCL\prec\CLSCLtwo$.

\noindent
Similarly, \MSCLeu\ axiomatises \MSCLu\ and is by Theorem~\ref{thm:indep2}.(iv) a proper subset of
the independent axiomatisation \CLeu\ of \CLSCLu. Hence, $\MSCLu\prec\CLSCLu$.
\end{proof}

\section{Discussion and conclusions}
\label{sec:conc}

We studied conditional logic and provided for its variants with distinguished 
constants \tr\ and \fa, and with or without \und, a detailed explanation of its `basic forms',
and included these variants in our framework of short-circuit logics. 
In this section we address two questions: which logic is defined by conditional logic when restricting to 
the full left-sequential connectives and negation, and what does it mean to define SCLs 
without the constants \tr\ and \fa. We end with some conclusions and remarks on related and future work.

\paragraph{Full left-sequential evaluation is commutative in CL -- consequences.}
Let $\CLthree(\fulland,\fullor)$ denote the extension of \CLthree\ with the connectives ${\fulland}$ and ${\fullor}$
and their defining axioms 
\[x\fulland y=(x\leftand y)\leftor(y\leftand x)\quad\text{ and }\quad
x\fullor y=\neg(\neg x\fulland \neg y).
\]
Then $\CLthree(\fulland,\fullor)$ incorporates Bochvar's 
\emph{three-valued strict logic}~\cite{Boc38}.
In~\cite{BergstraBR1995}, this logic is denoted as \SB\ and
it is proved that \SB\ is completely axiomatised by the following axioms:~\footnote{%
  The missing axiom 
  (S5) reads $x\fullimp y=\neg x\fullor y$, but the connective $\fullimp$ (full left-sequential implication)
  is irrelevant here. 
  We note that the axioms (S1){--}(S4),(S6)--(S9),(S11) are independent (according to \emph{Mace4}) and 
  that (S10) is a consequence of (S3),(S4),(S6)--(S9) (according to Prover9).
   }
\\[-4mm]
\begin{minipage}[t]{0.46\linewidth}\centering
\begin{Lalign}
\tag{S1}
\neg\tr&=\fa
\\
\tag{S2}
\neg\und&=\und
\\
\tag{S3}
\neg\neg x&=x
\\
\tag{S4}
\neg(x\fulland y)&=\neg x\fullor \neg y
\\
\tag{S6}
(x\fulland y)\fulland z
&=x\fulland (y\fulland z)
\end{Lalign}
\end{minipage}
\begin{minipage}[t]{0.62\linewidth}\centering
\begin{Lalign}
\tag{S7}
\tr\fulland x
&=x
\\
\tag{S8}
x\fullor (\neg x\fulland y)
&=x\fullor  y
\\
\tag{S9}
x\fulland y
&= y\fulland x
\\
\tag{S10}
x\fulland(y\fullor z)
&=(x\fulland y)\fullor(x\fulland z)
\\
\tag{S11}
\und\fulland x
&=\und
\end{Lalign}
\end{minipage}
\\[2mm]
With \emph{Prover9} it quickly follows that these axioms are consequences of $\CLthree(\fulland,\fullor)$. 
Moreover, in $\CLtwo(\fulland,\fullor)$ and $\CL(\fulland,\fullor)$, the corresponding version of the
axiomatisation of \SB, i.e. 
\[\text{(S1), (S3), (S4), (S6){--}(S9) and (S3), (S4), (S8){--}(S10), respectively,}
\]
is also derivable (the latter set implies (S6)).
The two-valued variant of \SB, say \SBtwo, satisfies 
\[\SBtwo\prec \textup{ ``propositional logic with $\neg$, $\fulland$ and $\fullor$''}\]
because the absorption law $x\fulland (x\fullor y)=x$ and its commutative variants (and duals) do not hold 
in \SBtwo\ (a counter-model is quickly found by \emph{Mace4}).

\paragraph{Short-circuit logics without constants.}
Obviously, all two-valued SCLs can be defined without the constants \tr\ and \fa.
In the case of \FSCL, this significantly reduces its axiomatisation,
a complete set of axioms is $\neg\neg x=x$, $x\leftor y=\neg(\neg x\leftand\neg y)$,
and $(x\leftand y)\leftand z=x\leftand(y\leftand z)$ (as noted in~\cite{PS18}).
All other axioms of \FSCL\ involve (sub)terms with only \fa-leaves (or with only \tr-leaves), 
such as $x\leftand\fa=\neg x\leftand\fa$, and these cannot be expressed in \FSCL\ without \tr\ and \fa. 
Hence, the full left-sequential connectives are not definable in \FSCL\ without \tr\ and \fa.

Let $\MSCL_0$ be defined as \MSCL, but without the constants \tr\ and \fa\ (which are not definable).
This is similar to the way CL is set up, leaving open the choice to adjoin \tr\ and \fa, 
and~\und. 
From the axiomatisation of \CLtwo\ in Table~\ref{tab:CL2} (Theorem~\ref{thm:indep2}.(ii)), 
a complete axiomatisation of $\MSCL_0$, say $\MSCLe_0$, is obtained by omitting 
the axiom~\eqref{Com} from those listed in Theorem~\ref{thm:indep2}.(i).\footnote{%
  $\MSCLe\vdash s=t\Rightarrow 
  \MSCLe_0\vdash s=t$ follows by induction on the length of derivations.}
Obviously, $\MSCLe_0$ is independent because the original axiomatisation is.
Moreover, full left-sequential conjunction is definable in $\MSCL_0$ by
\[x\fulland y=(x\leftor (y\leftand\neg y))\leftand y
\quad\text{or}\quad
x\fulland y=(x\leftand y)\leftor(y\leftand x).\]
We conclude that $\MSCL_0$ compared to \MSCL\ imposes no constraints other 
than the undefinability of \tr\ and \fa.

Finally, we note that $\MSCL_0\vdash(x\leftor \neg x)\leftand x = x,~(\neg x\leftor x)\leftand x = x$, 
the equations that are the counterparts of the axiom $\tr\leftand x =x$~\eqref{Tand}, 
so the question is indeed what would
be missed when using $\MSCL_0$ ``in practice''.
We also note that omitting \eqref{g} from  each of the two axiomatisations of
\CLtwo\ in Theorem~\ref{thm:indep} does not yield a complete axiomatisation of $\MSCL_0$: 
\emph{Mace4} 
finds a counter-model for $(\neg x\leftor x)\leftand x = x$, and for
$x\leftor \neg x=\neg x\leftor x$.

For the case of $\MSCLu_0$ (with \und\ and the axiom $\neg\und=\und$), similar observations apply.

\paragraph{Conclusions.} 
Starting from Guzmán and Squier's conditional logic, we introduced the two-valued
short-circuit logic \CLSCLtwo,
axiomatised by \CLtwo, and three-valued \CLSCLu, axiomatised by \CLthree. Then we analysed the commutativity
of the full left-sequential connectives in \CLthree\ and found that these and negation
define Bochvar’s three-valued strict logic. 
Finally, we investigated the independence of the equational axiomatisations of \CL, and 
paid attention to short-circuit logic without constants.

We conclude with two remarks. 
1) Three equations that hold in $\MSCL_0$ and express properties related to falsehood are 
$x\leftand \neg x=\neg x\leftand x$ and $(x\leftand \neg x)\leftor x=x$ and $(x\leftand \neg x) \leftand y =x\leftand \neg x$.
Another equation that relates to falsehood and does not hold in \MSCL, but does hold in \CL, is
\[
(x\leftand\neg x)\leftor (y\leftand\neg y)=(y\leftand\neg y)\leftor (x\leftand\neg x).
\]
Furthermore, $x\leftand \neg x = y\leftand\neg y$ is \emph{not} a consequence of \CLtwo\ (by \emph{Mace4}).
These facts support our conjecture that there exists no short-circuit logic without additional constants
in between \CLtwo\ and \SSCL\
(i.e., the SCL that represents propositional logic in the signature $\{{\leftand}, {\leftor}, \neg, \tr, \fa\}$).
We note that adding the axiom $x\leftand \neg x=y\leftand\neg y$ to \CLtwo\ (or $(x\leftand \neg x) \leftor y =y$,
or $x\leftand\fa=\fa$) yields \SSCL\ (by \emph{Prover9}). Of course, with $a\in A$ 
and $|A|>1$, the ad-hoc SCL defined by $\CLtwo\cup\{(a\leftand\neg a)\leftor x=x\}$ is in between \CLtwo\ and \SSCL.

2) The equation $(x\leftand\und)\leftor\und=\und$, which follows immediately from \CLu\ and \CLthree, 
can be seen as a concise characterisation of the difference between \CLu\ and $\MSCLu_0$,
and of the difference between \CLthree\ and \MSCLu.

\paragraph{Related work.} 
In \cite{Stau}, Staudt introduced evaluation trees and thereby
axiomatised \emph{Free Fully Evaluated Logic} (FFEL), a sublogic of \FSCL\ that has (only) 
full left-sequential connectives (and, like \FSCL, is immune to side effects):
\\[-5mm]
\begin{minipage}[t]{0.54\linewidth}\centering
\begin{Lalign}
\tag{FEL1}
\fa&=\neg\tr
\\
\tag{FEL2}
x\fullor y&=\neg(\neg x\fulland \neg y)\\
\tag{FEL3}
\neg\neg x&=x
\\
\tag{FEL4}
(x\fulland y)\fulland z&=x\fulland(y\fulland z)
\\
\tag{FEL5}
\tr\fulland x
&=x
\end{Lalign}
\end{minipage}
\hspace{-8mm}\begin{minipage}[t]{0.58\linewidth}\centering
\begin{Lalign}
\tag{FEL6}
x\fulland\tr&=x
\\
\tag{FEL7}
x\fulland \fa
&=\fa\fulland x
\\
\tag{FEL8}
x\fulland \fa
&= \neg x\fulland \fa
\\
\tag{FEL9}
(x\fulland\fa)\fullor y
&=(x\fullor\tr)\fulland y
\\
\tag{FEL10}
x\fullor(y\fulland \fa)&=x\fulland (y\fullor\tr)
\end{Lalign}
\end{minipage}
\\[2mm] 
These axioms are consequences of \SBtwo\ (according to Prover9), but FFEL is obviously weaker: 
$x\fulland y = y\fulland x$ does not hold.
(Note: a completeness result for FSCL was proved in~\cite{Stau}.)

\paragraph{Future work.}
In~\cite{BPS21}, we mentioned the following attractive challenge:
\emph{Find 
a convincing example that distinguishes $\MSCLu$ from Conditional Logic as
defined in Guzmán and Squier (1990).}

Our current research on ``fracterm calculus for partial meadows'' may well prove to be such a convincing example. 
A partial meadow is a field with a partial division function for which division by 0 is undefined.
Fracterm calculus for partial meadows aims to formalise elementary arithmetic including division by defining a 
three-valued sequential first-order logic, on top of a three-valued short-circuit logic with sequential, 
noncommutative connectives. In the Tarski semantics for this first-order logic we have in mind, 
conditional logic is sound and thus seems the right candidate to start from.

\appendix

\addcontentsline{toc}{section}{Appendix: A normalisation function for mem-basic forms}
\section*{Appendix: A normalisation function for mem-basic forms}

The two normalisation functions \baf()\ for two-valued basic forms and \membf()\
for mem-basic forms, both defined in~\cite{BP17},
can be easily lifted to the tree-valued setting.
 
Define the basic form function $\baf:\PSu\to\BFu$ and, given $Q,R\in\BFu$, 
the auxiliary function $[\tr\mapsto Q, \fa\mapsto R]:\BFu\to\BFu$
for which postfix notation $P[\tr\mapsto Q, \fa\mapsto R]$ is used, as follows:
\begin{align*}
&\text{$\baf(B)=B$ ~for $B\in\{\tr,\fa,\und\}$,}
&&\tr[\tr\mapsto Q, \fa\mapsto R]=Q,
\\
&\text{$\baf(a)=\tr\lef a\rig\fa$ ~for $a\in A$,}
&&\fa[\tr\mapsto Q, \fa\mapsto R]=R,
\\
&\baf(P \lef Q\rig R)= \baf(Q)[\tr\mapsto \baf(P), \fa\mapsto \baf(R)],
&&\und[\tr\mapsto Q, \fa\mapsto R]=\und,
\end{align*}
and $(P_1\lef a\rig P_2)[\tr\mapsto Q, \fa\mapsto R]=
P_1[\tr\mapsto Q, \fa\mapsto R]\lef a\rig P_2[\tr\mapsto Q,\fa\mapsto R]$.
\\[2mm]
From~\cite[App.A8]{BPS21}, it follows that for all $P\in\PSu$,
$\CPu\vdash P=\baf(P)$,
and for all $P,Q\in\PSu$, $\CPu\vdash P=Q$ if, and only if, $\baf(P)=\baf(Q)$.

Next, define the mem-basic form function $\membf:\PSu\to\PSu$ by $\membf(P)=\memf(\baf(P))$, 
and the auxiliary function 
  $\memf:\BFu\to\BFu$ by
  $\memf(B)=B$ ~for $B\in\{\tr,\fa,\und\}$
  and 
\(
\memf(P\lef a\rig Q)=\memf(\ell_a(P))\lef a\rig\memf(\ri_a(Q)),
\)
with $\ell_a()$ and $\ri_a()$ as in Definition~\ref{def:aux}.

It follows easily that $\membf()$ is a normalisation function yielding mem-basic forms:
the extension to \PSu\ of Theorem~5.9 in~\cite{BP17}, i.e.
\[\text{For all $P\in\PSu, \CPmem \vdash P = \membf(P)$,}\]
and of the supporting lemmas is trivial: all inductive proofs in~\cite{BP17}
require one additional, trivial base case (for \und).

\end{document}